%% file: main.tex
\documentclass{colt2021}
\usepackage{macro}
\usepackage{amsfonts}
\usepackage{xcolor}
\usepackage{bbm}
\usepackage{cleveref}

\newcommand{\Apx}{\mathit{Apx}}

\newcommand{\epsilonsub}[1]{\varepsilon_{\textnormal{#1}}}
\newcommand{\F}{\mathsf{F}}
\renewcommand{\P}{\mathbb{P}}
\renewcommand{\tilde}{\widetilde}

\newcommand{\R}{\mathbb{R}}
\renewcommand{\epsilon}{\varepsilon}
\coltauthor{\Name{Praneeth Kacham} \Email{pkacham@cs.cmu.edu} \and \Name{David P. Woodruff} \Email{dwoodruf@cs.cmu.edu}\\ \addr Computer Science Department \\ Carnegie Mellon University}
\date{}
\title{Reduced-Rank Regression with Operator Norm Error}
\usepackage{times}
\usepackage{inconsolata}
\begin{document}
\maketitle
\begin{abstract}
     A common data analysis task is the reduced-rank regression problem:
        $$\min_{\textrm{rank-}k \ X} \|AX-B\|,$$
    where $A \in \mathbb{R}^{n \times c}$ and $B \in \mathbb{R}^{n \times d}$ are given large matrices and $\|\cdot\|$ is some norm. Here the unknown matrix $X \in \mathbb{R}^{c \times d}$
    is constrained to be of rank $k$ as it results in a significant parameter reduction of the solution when $c$ and $d$ are large.
    In the case of Frobenius norm error, there is a standard closed form solution to this problem and
    a fast algorithm to find a $(1+\varepsilon)$-approximate solution. However, for the important case of 
    operator norm error, no closed form solution is known and the fastest known algorithms take singular
    value decomposition time. 
    
    We give the first randomized algorithms
    for this problem running in time 
    $$(\nnz{A} + \nnz{B} + c^2) \cdot k/\varepsilon^{1.5} + (n+d)k^2/\epsilon + c^{\omega},$$
    up to a polylogarithmic
    factor involving  condition numbers, matrix dimensions, and dependence on $1/\varepsilon$. Here 
    $\nnz{M}$ denotes the number of nonzero entries of a matrix $M$, and
    $\omega$ is the exponent of matrix multiplication. 
    As both (1) spectral low rank approximation ($A = B$) and (2) linear system solving
    ($n = c$ and $d = 1$) are special cases, 
    our time cannot be improved by more than a $1/\varepsilon$ factor (up to polylogarithmic factors)
    without a major breakthrough in linear algebra. Interestingly, known techniques for low rank approximation,
    such as alternating minimization or sketch-and-solve, provably fail for this problem. 
    Instead, our algorithm uses an existential characterization of a 
    solution, together with Krylov methods, low degree polynomial approximation, and sketching-based preconditioning.
\end{abstract}
% \twocolumn
\section{Introduction}
Given an $n \times c$ matrix $A$, an $n \times d$ matrix $B$, and an integer parameter $k$, the reduced-rank regression problem asks to solve for a rank at most $k$ matrix $X \in \mathbb{R}^{c \times d}$ for which
$\|AX-B\|$ is minimized in some norm. A standard motivation is that by 
constraining $X$ to have rank at most $k$, the solution $X$ can be
represented using only $(c+d)k$ parameters rather than $c\cdot d$ parameters. Another important motivation is
that the rank constraint provides regularization on the solution, which often leads to better generalization.
Yet another motivation is that the solution $X$ can be explained by at most $k$ latent factors, and one can
try to interpret the latent factors, plot them \citep{braak1994biplots}, and so on. This is commonly done in ecology, where reduced-rank regression
is known as redundancy analysis \citep{legendre1999distance}, and is a type of ordination method \citep{kobak2019sparse}. For a survey, we refer the reader to the textbook by \citet{velu2013multivariate} devoted to reduced-rank regression. 

%There has also been a long line of work on Column Sampling based low-rank matrix approximation (see \cite{boutsidis2011topics}, and the references therein). To approximate a matrix $B \in \R^{n \times d}$, we first sample columns of $B$ uniformly at random or using other importance sampling procedures, to obtain a matrix $A = BR$ where $R$ selects a subset of sampled columns of $B$. We then solve the problem $\min_{\text{rank-}k\ X}\|AX-B\|$ for some norm $\| \cdot \|$. For operator norm, it is known that there is an algorithm \cite[Theorem~33]{boutsidis2011topics} that samples $r$ columns of $B$ to form a matrix $A = BR$ for which
%\begin{equation*}
%	\min_{\text{rank-}k\ X} \opnorm{AX-B} = O\left(\sqrt{{n}/{r}}\right)\opnorm{B - [B]_k}.
%\end{equation*}
%Though the above guarantee is weaker than the one given by Block Krylov Iteration\citep{muscomusco}, it is interesting as a low rank approximation is constructed using an actual subset of columns of the matrix $B$. No input sparsity time, efficient algorithms were previously known to solve the problem $\min_{\text{rank-}k\ X}\opnorm{AX-B}$ which is what we address in this work.

The $\min_{\text{rank-}k\ X}\|AX-B\|$ problem is only known to have a closed form solution when the error
measure is the Frobenius norm. In this case, the solution is given by 
$X = A^{+}[AA^{+} B]_k$ (see, e.g., \cite{f07}). Here for a matrix $M$, $[M]_k$ denotes the best rank $k$ approximation for $M$ in Frobenius norm and $M^+$ denotes the Moore-Penrose pseudo-inverse. This has a natural geometric interpretation - project each of the columns of $B$ onto the column span of $A$
and find the best rank-$k$ approximation to the projected matrix. By the Pythagorean
theorem, one can show there is no loss in this approach, as the optimal cost decomposes into the sum of squared
distances of columns of $B$ to the column span of $A$ followed by the best rank-$k$ approximation to the projected
matrix inside of the column span of $A$. 

In a number of applications, the Frobenius norm is not the right measure. For example, in cancer genetics more robust versions are desired, and versions based on the sum of Euclidean lengths instead of the sum of squared
Euclidean lengths are sometimes used \citep{she2017robust}. Still, in other applications, the operator norm error solution may give a solution of much better quality. Indeed,
if $B$ has a heavy tail of singular values, as is common for data analysis and
learning applications, then it has no good rank-$k$ approximation,
much less one in the column span of $A$, and consequently, outputting an $X'$
with $\|AX'-B\|_\F^2 \leq (1+\epsilon)\|AX_\F-B\|_\F^2$, where $X_\F$ is the optimal
Frobenius norm solution, may be meaningless as one could just set $X' = 0$. Indeed,
this is sometimes a motivation (see, e.g., \cite{muscomusco}) for the low rank approximation problem with operator norm error, which is a special case of our problem when $A = B$, and a number of works \citep{halko2011finding,jlsw20,karnin2015online,szlam2014implementation} suggest considering operator norm error in certain contexts.  

It is tempting to think that the Frobenius norm solution 
holds also for other unitarily invariant norms, such as the operator 
norm. However, one can show this is not the case. 
Indeed, let $X_\F$ be the solution to $\min_{\textrm{rank-}k \ X} \frnorm{AX-B}$. It was shown by \citet{boutsidis2011topics} that this is a $\sqrt{2}$-approximation, namely, that $\opnorm{AX_\F-B} \leq \sqrt{2} \cdot \opt$ where $\opt = \min_{\text{rank-}k\ X}\opnorm{AX-B}$. Unfortunately, the $\sqrt{2}$ factor is tight and there are instances where the Frobenius norm solution really does give at best a $\sqrt{2}$-approximation. Suppose, for example\footnote{We thank Ankur Moitra for pointing out this example to us.} $$A = 
 \begin{bmatrix}
 0&  0\\
 1 & 0\\
 0 & 1
 \end{bmatrix}, \text{ and }
 B = \begin{bmatrix}
 1&  0\\
 1 & 0\\
 0 & 1+\gamma
 \end{bmatrix}.
$$
For the problem $\min_{\text{rank-}1\ X} \|AX-B\|_\F$, the optimum solution is 
$$
X_{\F} = \begin{bmatrix}0 & 0 \\ 0 & 1+\gamma \end{bmatrix}\text{, with } AX_{\F}-B = -\begin{bmatrix} 1 & 0 \\ 1& 0 \\ 0 & 0 \end{bmatrix}
$$
and thus, $\|AX_{\F}-B\|_2 = \sqrt{2}$. On the other hand, for 
$$
X = \begin{bmatrix}1 & 0 \\ 0 & 0\end{bmatrix},\ AX- B = -\begin{bmatrix}1 & 0 \\ 0 & 0 \\ 0 & 1+\gamma \end{bmatrix},
$$
and so $\|AX-B\|_2 = (1+\gamma)$. As $\gamma \rightarrow 0$, the approximation factor becomes arbitrarily close to $\sqrt{2}$.

We note that the reduced-rank regression problem in operator norm is non-convex in $X$ due to the rank constraint, and it is not even clear this problem can be solved in polynomial time. Of the few techniques that are known for rank-constrained optimization, they do not apply here. One common method is alternating minimization, writing the problem above as $\min_{U, V} \|AUV-B\|_2$, where $U \in \mathbb{R}^{n \times k}$ and $V \in \mathbb{R}^{k \times d}$. The idea is to fix $U$, then solve for $V$, then fix $V$ and solve for $U$, and repeat. When $U$ is fixed, then $V = (AU)^{+} B$ is the optimum, and when $V$ is fixed, the solution turns out to be $U = A^+ B V^+$, though this is not as obvious, see (1.3) in \cite{m07}, taking $p \rightarrow \infty$, for a proof. It turns out if one initializes with the Frobenius norm solution $U, V$, then each of these operations does not change $U$ or $V$, and so by the example above, alternating minimization gives at best a $\sqrt{2}$-approximation. Other techniques include sketching to a small problem,
and solving the small problem in the sketch space; sketches are well-known not to 
apply to operator norm low rank approximation problems, motivating the first open
question in \cite{sketching-dw}. 

This issue of polynomial time solvability was raised in the control theory literature by \citet{sou2012generalized}, where a $(1+\varepsilon)$-approximation was obtained, but the the time required to find the solution was at least the time to perform a singular value decomposition (SVD) on matrices $A$ and $B$, which is prohibitive for large $n,c,$ and $d$. This is a common setting of parameters and indeed, one of the motivations for constraining $X$ to have rank at most $k$ in the first place. This motivates the question:

\begin{center}
``\textit{Are there fast algorithms for reduced-rank regression with operator norm error?}''
\end{center}

\subsection{Main Result} 
We answer the question above by designing a new randomized algorithm running
in time $$O\left(\left(\frac{\nnz{B} \cdot k}{\varepsilon} +\frac{\nnz{A}\cdot k}{\varepsilon^{1.5}} +\frac{c^2k}{\varepsilon^{1.5}} + \frac{(n+d)k^2}{\varepsilon} \right)\cdot \textnormal{polylog}(\kappa(B),n, d,k,1/\varepsilon)+ c^{\omega}\right).$$ Here, $\kappa(B)$ denotes $\sigma_1(B)/\sigma_{k+1}(B)$. 
This significantly
improves over Sou and Rantzer's polynomial time result, which takes $\Omega(nd^2 + nc^2)$ time.

We note that spectral
low rank approximation is a special case in which $A = B$, and the best known upper
bound is $O(\nnz{A} \cdot k/\sqrt{\varepsilon})$ for this problem, up to logarithmic factors
\cite{muscomusco}. A major open question in randomized numerical linear algebra
is to improve this bound (see, e.g., Open Question 1 of \cite{sketching-dw}), or show
that it is not possible. We note that for $k = 1$, in the matrix-vector query model,
$\Omega(1/\sqrt{\varepsilon})$ queries is known to be required if a slightly stronger guarantee
than spectral low rank approximation is desired, even for adaptive
algorithms \cite{bhsw20, simchowitz2018tight}. 
Another important point is that when $n=c$ and $d = 1$, this
is just the time to solve an arbitrary linear system, for which the best known
time is $c^{\omega}$. Improving either spectral low rank approximation or linear
system solving is a major open question, and barring that, our algorithm is optimal
up to a $1/\varepsilon$ factor and polylogarithmic factors involving
matrix dimensions and condition numbers. 

\subsection{Our Techniques}
Throughout the paper, let $\opt := \inf_{\text{rank-}k\ X}\opnorm{AX-B}$, $\beta$ be such that $(1+\varepsilon)\opt \le \beta \le (1+2\varepsilon)\opt$, and let $\Delta := \T{B}(I-AA^+)B$. The work of \citet{sou2012generalized} shows that $X_\beta = A^+ [AA^+ B(\beta^2I - \Delta)^{-1/2}]_k(\beta^2I - \Delta)^{1/2}$ satisfies $\opnorm{AX - B} < \beta$. For completeness, we give a short proof of this
fact in this paper. It is not a priori clear how to extract a solution from this expression,
while multiplying out all of the matrices, computing an inverse square root, and taking an SVD would take a prohibitive amount of time. This is essentially the algorithm
of \cite{sou2012generalized}.

%Throughout the paper, let $\opt := \min_{\text{rank-}k\ X}\opnorm{AX-B}$ and $\beta$ be such that $(1+\varepsilon)\opt \le \beta \le (1+2\varepsilon)\opt$. The work of \cite{sou2012generalized}
%implies an \emph{existential result}, namely that there exists a solution $X$ of cost $\beta$ such that $AX$ is  in the column span of the rank $k$ matrix
%$[AA^{+} B (\beta^2 I - B^T(I-AA^{+})B)^{-1/2}]_k,$
%where for a matrix $C$, the matrix $[C]_k$ denotes its
%best rank-$k$ approximation, as given by the SVD. 
%For completeness, we give a short proof of this
%fact in this paper. It is not a priori clear how to extract a solution from this expression,
%while multiplying out all of the matrices, computing an inverse square root, and taking an SVD would take a prohibitive amount of time. This is essentially the algorithm
%of \cite{sou2012generalized}.

We instead show that not only the best rank $k$ approximation of the matrix $AA^+ B(\beta^2I - \Delta)^{-1/2}$, but even a $1+\varepsilon$ approximation in spectral norm yields an overall solution of cost at most  $\beta(1+O(\varepsilon))$. To obtain such a $1+\varepsilon$ approximation, we next try to apply the iterative method of 
\cite{muscomusco} which computes the Krylov matrix
$K =  [C \cdot G, (C\T{C})\cdot C \cdot G, (C\T{C})^2\cdot C \cdot G, \ldots, (C\T{C})^{(q-1)/2}\cdot C \cdot G]$
where $G$ is a Gaussian matrix with $k$ columns, $q = O(\log(d/\varepsilon) \sqrt{1/\varepsilon})$ is an odd integer, and $C = AA^{+} B (\beta^2 I - \Delta)^{-1/2}$.
The first problem with this approach is that we have to compute the matrix vector product $CG$ and to do this, in each
iteration we need to (1) multiply by the square root of an inverse (multiplication by $(\beta^2 I - \Delta)^{-1/2})$, and then (2) project onto the column span of $A$ (multiplication by $AA^{+}$). 

Computing exact matrix-vector products with the matrices $AA^+$ and $(\beta^2I - \Delta)^{-1/2}$, is slow when $c,d$ are large, and finding the matrices $AA^+$ and $(\beta^2I - \Delta)^{-1/2}$ takes at least $\Omega(nc^2+\nnz{B}\cdot c + d^{\omega})$ time. To avoid such a running time, we show that the Block Krylov Iteration algorithm of \citet{muscomusco} works even with approximate matrix-vector products i.e., we only need algorithms to compute vectors $C \circ v$ and $\T{C} \circ v'$ for arbitrary vectors $v,v'$ such that $\opnorm{C \circ v - Cv}$ and $\opnorm{\T{C}\circ v' - \T{C}v'}$ are small. Here and throughout the paper, we use the notation $M \circ v$ to denote an approximation to the matrix-vector product $Mv$. 

An important idea of \citet{muscomusco} is that the Krylov matrix $K$ spans a rank $k$ matrix $p(C)G = \sum_{\text{odd\ }i \le q} p_i (C\T{C})^{(i-1)/2}G$, where $p$ is a polynomial, such that projecting the columns of the matrix $C$ onto the column span of $p(C)G$ gives a good rank $k$ approximation. To prove that the algorithm works even with approximate matrix-vector products, we first show that the approximations computed to matrices $(C\T{C})^{(i-1)/2}CG$ for $i=1,\ldots,q$ are good enough to imply that the approximate Krylov matrix $K'$ spans a matrix $\Apx$ that is close to the matrix $p(C)G$ in Frobenius norm. To then conclude that the column space of $\Apx$ is also a good subspace to project the matrix $C$ onto, we need to show that $(\Apx)(\Apx)^+ \approx (p(C)G)(p(C)G)^+$. We prove a simple lemma that shows if $\frnorm{p(C)G - \Apx}$ is small, and $p(C)G$ has a good condition number, and so then $\opnorm{(p(C)G)(p(C)G)^+ - (\Apx)(\Apx)^+}$ is small. Crucially, as $G$ is a Gaussian matrix that has, with good probability a good condition number, we only have to bound $\sigma_1(p(C))/\sigma_k(p(C))$ to obtain a bound on the condition number of $p(C)G$. Using several properties of Chebyshev polynomials used to define the polynomial $p(x)$, we show that $\sigma_1(p(C))/\sigma_{k}(p(C))$ can be bounded in terms of $\kappa = \sigma_1(C)/\sigma_{k+1}(C)$, which finally shows that the $k$-dimensional column span of $\Apx$ is also a good subspace to project the columns of $C$.

As the parameters of the polynomial $p(x)$ are unknown, we cannot actually compute the matrix $\Apx$ and then project $C$ onto the column span. But using the fact that $K'$ spans $\Apx$, we can conclude, similarly to the arguments of \cite{muscomusco}, that the best rank $k$ Frobenius norm approximation of $C$ in the span of $K'$ is a good rank $k$ approximation to $C$. Using the oracle to compute approximate matrix-vector products with the matrix $C$, we recover a $1+\varepsilon$ approximation to the best rank $k$ Frobenius norm approximation of $C$ inside the span of $K'$, which we then show is a $1+\varepsilon$ approximation to a spectral norm low rank approximation of matrix $C$. Our analysis that the Block Krylov Iteration algorithm works with approximate matrix-vector products could help justify why the Block Krylov Iteration algorithm works well when using finite precision arithmetic rather than exact arithmetic. Our results address the comments of \cite{mms-lancosz} about the stability of block Lanczos based methods for problems such as low rank approximation. Though several analyses of the noisy power method have been done previously \citep{balcan-du-wang-yu,hardt-price,hardt-roth}, where each intermediate computation is corrupted by Gaussian noise, we are not aware of an analysis that works for worst case corruption. Also, previous work bounds the amount of Gaussian noise that can be added in terms of a gap between $\sigma_k$ and $\sigma_{k+1}$, which can be $0$, and would not work for our analysis.

We return to the task at hand, i.e., of computing a low rank approximation of $AA^+ B(\beta^2I - \Delta)^{-1/2}$. We show that we can replace the matrix $(\beta^2I - \Delta)^{-1/2}$ with the matrix $(1/\beta) {r}(\Delta/\beta^2)$, where $ {r}(x)$ is a polynomial of degree $\tilde{O}(1/\sqrt{\varepsilon})$, using polynomial approximation techniques based on Chebyshev polynomials (see, e.g., \cite{sachdeva-vishnoi} and the references therein). Here we crucially use the fact that $(1+2\varepsilon)\opt \ge \beta \ge (1+\varepsilon)\opt \ge (1+\varepsilon)\opnorm{(I-AA^+)B}$ to lower bound the minimum singular value of the matrix $(\beta^2I - \Delta)$, thereby obtaining an upper bound on the number of terms required to approximate $(I-(\Delta/\beta^2))^{-1/2}$ with a Taylor series. Then we replace each monomial in the Taylor series with a low degree polynomial approximation to construct a polynomial $ {r}(x)$. The replacement of $(\beta^2I - \Delta)^{-1/2}$ with the matrix $ {r}(\Delta/\beta^2)$ is done as we can give very fast algorithms to approximately multiply a vector with the matrix $ {r}(\Delta/\beta^2)$, as discussed below. 

Let $\mathcal{M'} = AA^+ B\cdot  {r}(\Delta/\beta^2)$. Recall $\Delta = \T{B}(I-AA^+)B$. To approximate the matrix-vector product $\Delta u$ for an arbitrary vector $u$, we need only approximate $\T{B}AA^+ Bu$, since $\T{B}Bu$ can be computed exactly in $\nnz{B}$ time. For computing an approximation to $AA^+(Bu)$, we use fast sketching-based preconditioning methods for linear regression, which show given an arbitrary vector $b$ and accuracy parameter $\epsilonsub{reg}$ how to find an $x$ for which $\opnorm{Ax - AA^+ b} \le \epsilonsub{reg}\opnorm{(I-AA^+)b}$ in time $O((\nnz{A} + c^2)\log(1/\epsilonsub{reg}) + c^\omega)$, where $\omega \approx 2.376$ is the exponent of matrix multiplication \citep{cw13,mm13,nn13}. We note that we only need to pay the $c^{\omega}$ time once
to compute a preconditioner, after which each regression problem takes $O((\nnz{A} + c^2)\log(1/\epsilonsub{reg}))$ time. This algorithm to approximately compute $\Delta u$ for an arbitrary vector $u$ is extended to approximate $ {r}(\Delta/\beta^2) \cdot v$ for an arbitrary $v$. After approximating the product $ {r}(\Delta/\beta^2) \cdot v$ with a vector $y$, we approximate the vector $AA^+ By$ again using the sketching-based preconditioning methods for linear regression.

Similarly we also give an algorithm to approximate $\T{\mathcal{M'}}v'$ for an arbitrary vector $v'$. Thus, as discussed above, we can obtain using a Block Krylov algorithm, a matrix $Z$ with orthonormal columns for which
$
	\opnorm{Z\T{Z}\mathcal{M'} - \mathcal{M'}} \le (1+\varepsilon)\sigma_{k+1}(\mathcal{M'})
$
and then conclude that
\begin{equation*}
	\opnorm{AA^+ Z(AA^+ Z)^+ B - B} \le (1+O(\varepsilon))\beta = (1+O(\varepsilon))\opt
\end{equation*}
and that the rank $k$ matrix $X = A^+ Z(AA^+ Z)^+ B$ is a $1+O(\varepsilon)$ approximation for the problem $\min_{\text{rank-}k}\opnorm{AX-B}$. 

The time complexity of our algorithm depends logarithmically on $\kappa(B) = \sigma_1(B)/\sigma_{k+1}(B)$ and $\kappa(AA^+ B) = \sigma_1(AA^+ B)/\sigma_{k+1}(AA^+ B)$. We show that if $\tilde{B} = B + \alpha G\T{F}$ where $G$ is an $n \times (k+1)$ random Gaussian matrix and $\T{F}$ has $k+1$ orthonormal rows, then for a suitable value of $\alpha$, the condition number $\kappa(AA^+ \tilde{B}) \le (Cn/\varepsilon)\kappa(B)$ for a constant $C$. We also show that a $1+\varepsilon$ approximation for reduced rank regression computed using the matrix $\tilde{B}$ is a $1+O(\varepsilon)$ approximation for reduced rank regression on matrix $B$, thus removing the  dependence on $\kappa(AA^+ B)$. Note that matrix-vector products with $\tilde{B}$ can be computed in $\nnz{B} + (n+d)k$ time.

Our final dependence on $\varepsilon$ in the running time is $1/\varepsilon^{3/2}$, ignoring polylogarithmic factors, where a factor of $1/\sqrt{\varepsilon}$ is from the number of iterations in the Block Krylov Iteration algorithm of \citet{muscomusco}, a factor of $1/\sqrt{\varepsilon}$ is from the degree of the polynomial $ {r}(x)$, which is used as a proxy for the matrix $(\beta^2I - \Delta)^{-1/2}$ with a matrix $ {r}(\Delta/\beta^2)$, and a factor of $1/\sqrt{\varepsilon}$ is due to the running time of high-precision regression methods based on the accuracy with which the approximate matrix products need to be computed.

\section{Notation and Preliminaries}
For a matrix $M$, $\nnz{M}$ denotes the number of nonzero entries in $M$. We refer to the Singular Value Decomposition (SVD) with only nonzero singular values as the ``thin'' SVD. Given an arbitrary matrix $M$, $\text{colpsan}(M)$ denotes the subspace spanned by the columns of $M$, and the matrix $M^{+}$ denotes the Moore-Penrose pseudo-inverse of matrix $M$. Given a subspace $V$, the matrix $\P_V$ denotes the projection onto the subspace $V$. Therefore $\P_Vu = \argmin_{v \in V}\opnorm{u - v}$ for all vectors $u$. Given a matrix $M$, we use $\P_M$ to denote $\P_{\text{colspan}(M)}$.

For a matrix $M$, the Frobenius norm $(\sum_{i,j} M_{i,j}^2)^{1/2}$ is denoted by $\frnorm{M}$ and the operator norm (or spectral norm) $\sup_{x} \opnorm{Mx}/\opnorm{x}$ is denoted by $\opnorm{M}$. For a square matrix $M$, \text{tr}($M$) denotes the sum of diagonal entries. For matrices $M$ and $M'$ of the same dimensions, $\langle M, M' \rangle$ denotes $\text{tr}(\T{M}M') = \sum_{i,j}M_{i,j}M'_{i,j}$. We use the following standard facts repeatedly throughout the paper: for any matrix $M$, (1) $\opnorm{M} \le \frnorm{M}$, (2) $\frnorm{M} \le \sqrt{\text{rank}(M)}\opnorm{M}$ and (3) $\P_M = MM^{+}$. For any matrices $A,B$ and $C$, (i) $\text{tr}(ABC) = \text{tr}(BCA)$, (ii) $\frnorm{ABC} \le \opnorm{A}\frnorm{B}\opnorm{C}$ and (iii) $\langle A, B \rangle \le \frnorm{A}\frnorm{B}$.

For a symmetric matrix $M$, define $\text{psd}(M)$ to be the closest positive semi-definite matrix to $M$ in Frobenius norm. It can be shown that if $M = \sum_i \lambda_i v_i\T{v_i}$, then $\text{psd}(M) = \sum_{i : \lambda_i \ge 0}\lambda_i v_i\T{v_i}$.

\textbf{Weyl's Inequality.} For matrices $A$ and $B$, Weyl's inequality gives that 
$
 \sigma_{i+j-1}(A + B) \le \sigma_i(A) + \sigma_j(B)   
$
for all $i$ and $j$. In particular, if $\opnorm{A - B} \le \varepsilon$, $|\sigma_i(A) - \sigma_i(B)| \le \varepsilon$ for all $i$.

\textbf{Polynomials and Matrices.} Let $ {p}(x) = \sum_{i=0}^d p_ix^i$ be a degree $d$ polynomial. We define $\| {p}\|_1 := \sum_i |p_i|$ to be the sum of absolute values of the coefficients of the polynomial $p(x)$. Given $A \in \R^{n \times d}$, let $A = U\Sigma\T{V}$ be the singular value decomposition of $A$ with $\Sigma \in \R^{n \times d}$. Define $ {p}(A) := Up(\Sigma)\T{V}$ where $ {p}(\Sigma)$ is the matrix with main diagonal entries $ {p}(\sigma_1),\ldots, {p}(\sigma_d)$. It is easy to check that the singular values of $ {p}(A)$ are equal to $| {p}(\sigma_1)|,\ldots,| {p}(\sigma_d)|$.

\textbf{Singular Value Excess.} Let $A \in \R^{n \times d}$ with $n \ge d$ be an arbitrary matrix. Let $\sigma_1 \ge \sigma_2 \ge \cdots \ge \sigma_d \ge 0$ be the singular values of matrix $A$. The Singular Value Excess of matrix $A$, denoted by $\text{sve}(A)$, is defined as the number of singular values of matrix $A$ that are greater than or equal to $1$ i.e.,
\begin{equation*}
	\text{sve}(A) = |\set{i \in [d]\,|\, \sigma_i \ge 1}|.
\end{equation*}
As eigenvalues of matrix $I-\T{A}A$ are $1-\sigma_1^2 \le \cdots \le 1-\sigma_d^2$, $\text{sve}(A)$ is equal to the number of non-positive eigenvalues of the matrix $I-\T{A}A$. For any symmetric matrix $M$, let $k^-(M)$ denote the number of non-positive eigenvalues of the matrix $M$. For any matrix $A$,
$
	\text{sve}(A) = k^-(I-\T{A}A).
$

\textbf{Sketching Based Preconditioning for High-Precision Regression.} Given a matrix $A \in \R^{n \times c}$ and a vector $b \in \R^{n}$, we use fast sketching based preconditioning methods given by the following theorem to obtain a $(1+\varepsilon)$ approximation to the problem $\min_x \opnorm{Ax - b}$. See \cite{sketching-dw} and references therein for more background.
\begin{theorem}[High Precision Regression/Approximate Projections]
	\label{thm:high-precision-regression}
	Given a matrix $A \in \Real^{n \times c}$ and a vector $b \in \Real^n$, we can compute a vector $x$ in time $O((\nnz{A}+c^2)\log(1/\varepsilon) + c^\omega)$ that satisfies
	$
		\opnorm{Ax - b}^2 \le (1+\varepsilon)\opnorm{AA^+ b - b}^2.
	$
By the Pythagorean theorem, the vector $x$ obtained satisfies
$
	\opnorm{AA^+ b - Ax}^2 \le \varepsilon\opnorm{AA^+ b - b}^2.
$
\end{theorem}
We have to pay $c^\omega$ only once to compute a preconditioner. Thereafter, every regression problem can be solved in time $O((\nnz{A}+c^2)\log(1/\varepsilon))$. Throughout the paper, we use \textsc{HighPrecisionRegression}$(A,b, \epsilon)$ to denote the algorithm implied by Theorem~1. We extend the notation to compute approximate projections of each of the columns of matrix $B$, instead of just a single vector $b$, onto the column space of $A$.

\textbf{Low Rank Approximation(LRA).} Let $A \in \Real^{n \times c}$ and $A = U\Sigma\T{V}$ be its ``thin'' Singular Value Decomposition, where $\T{U}U = I$, $\T{V}V=I$ and $\Sigma = \text{diag}(\sigma_1,\ldots,\sigma_{\text{rank}(A)})$ with $\sigma_1 \ge \sigma_2 \ge \cdots\ge \sigma_{\text{rank}(A)} > 0$. For any $k \le \text{rank}(A)$, we define
$
    [A]_k := \sum_{i=1}^k \sigma_iU_{*i}(\T{V})_{i*},
$
where $U_{*i}$ denotes the $i$-th column of matrix $U$ and $\T{V}_{i*}$ denotes the $i$-th row of matrix $\T{V}$. The matrix $[A]_k$ optimally solves the problems $\min_{\text{rank-}k\ X}\frnorm{A - X}$ and $\min_{\text{rank-}k\ X}\opnorm{A - X}$. As computing $[A]_k$ exactly is expensive, we use the Block Krylov Iteration algorithm of \cite{muscomusco} to obtain a matrix $Z \in \Real^{n \times k}$ for which $Z\T{Z}A$ is a good solution to the Frobenius norm and spectral norm low rank approximation problems.
\begin{theorem}[\citet{muscomusco}]
    \label{thm:musco-musco}
    Given a matrix $A \in \Real^{n \times d}$ such that the products $Av \in \Real^{n}$ and $\T{A}v' \in \Real^d$ can be computed in time $T$ for any vectors $v \in \Real^d$ and $v' \in \Real^n$, the Block Krylov Iteration algorithm runs in time
    \begin{equation*}
        O\left(T\frac{k\log d}{\varepsilon^{1/2}} + \frac{nk^2\log^2(d)}{\varepsilon} + \frac{k^3\log^3(d)}{\varepsilon^{3/2}}\right) 
    \end{equation*}
    and returns a matrix $Z \in \Real^{n \times k}$ with orthonormal columns for which  
    $$
        \opnorm{A - Z\T{Z}A} \le (1 + \varepsilon)\opnorm{A - [A]_k}\ \text{and}\ \frnorm{A - Z\T{Z}A} \le (1 + \varepsilon)\frnorm{A - [A]_k}.$$
\end{theorem}

\textbf{Frobenius Norm Reduced-Rank Regression.} As discussed in the introduction, there is a closed form solution to the reduced-rank Frobenius norm regression problem.
\begin{lemma}[Lemma~4.1 of \cite{sketching-dw}, Lemma~2 of \cite{muscomusco}]\label{lma:frobenius-norm-rrr}
Given matrices $A \in \R^{n \times c}$, $B \in \R^{n \times d}$, and a rank parameter $k \le c$, let matrix $Q$ denote an orthonormal basis for the column span of $A$. Then
$
    \min_{\text{rank-}k\ X} \frnorm{AX - B} = \frnorm{Q[\T{Q}B]_k - B} = \frnorm{[AA^+ B]_k - B}.
$
If $\bar{U}\bar{\Sigma}^2\T{\bar{U}}$ is the SVD of $\T{Q}A\T{A}Q$, and $\bar{U}_k$ denotes the first $k$ columns of $\bar{U}$, then $[\T{Q}B]_k = \bar{U}_k\T{\bar{U}_k}\T{Q}B$, and therefore
\begin{equation*}
    \min_{\text{rank-}k\ X}\frnorm{AX - B} = \frnorm{Q[\T{Q}B]_k - B} = \frnorm{(Q\bar{U}_k)\T{(Q\bar{U}_k)}B - B}.
\end{equation*}
\end{lemma}

\textbf{Chebyshev Polynomials.} The Chebyshev polynomials are defined as 
\begin{align*}
    T_0(x) = 1,\ 
    T_1(x) = x\ \text{and}\ 
    T_i(x) = 2xT_{i-1}(x) - T_{i-2}(x)
\end{align*}
for all $i \ge 2$. Thus $T_i(x)$ is a polynomial of degree $i$. It can be shown that if $i$ is odd, then $T_i(x)$ has only odd degree monomials. Chebyshev polynomial $T_i$ has the property that $\|T_i\|_1 \le (1+\sqrt{2})^i$ for all $i$. See \cite{muscomusco} for more properties of Chebyshev polynomials.

\section{Previous work}
\label{sec:related-work}
Let $A \in \R^{n \times c}$ be a matrix and $U\Sigma V^{\mathsf{T}}$ be the ``thin'' SVD of $A$, where $U$ is an orthonormal basis for the column space of $A$. Note that the projection matrix onto the column space of $A$ is given by $AA^+ = U\T{U}$.
The first algorithm to solve $\min_{\text{rank-}k\ X}\opnorm{AX - B}$ 
was by \citet{sou2012generalized}. They consider the following problem:
\begin{align}
    \text{minimize } &\rank(X)  \nonumber\\ \text{such that } &\opnorm{AX - \B} < 1. \label{eqn:sr-problem}
\end{align}
As multiplying a matrix with a projection matrix does not increase the operator norm, we have that  $\opnorm{AX-B} \ge \opnorm{(I-AA^+)(AX-B)} = \opnorm{(I-AA^+)B}$. Thus the problem is feasible only when $\opnorm{(I-AA^+)B} = \opnorm{(I-U\T{U})B} < 1$. The following theorem characterizes the solution for~\eqref{eqn:sr-problem}.
\begin{theorem}[\citet{sou2012generalized}]
    \label{thm:sou-rantzer}
Given matrices $A \in \R^{n \times c}$ and a matrix $B \in \R^{n \times d}$, if there is a matrix $Y$ such that $\opnorm{AY - B} < 1$, then the optimum value of \eqref{eqn:sr-problem} is sve($B$) where sve($B$) denotes the number of singular values of $B$ that are greater than or equal to $1$.
\end{theorem}
For an arbitrary $s > 0$, consider the problem \eqref{eqn:sr-problem} with matrices $A/s$ and $B/s$. The problem is feasible if and only if $\opnorm{(I-U\T{U})(B/s)} < 1$, i.e., if and only if $\opnorm{(I-U\T{U})B}< s$. Suppose $s$ is such that $s > \opnorm{(I-U\T{U})B}$. Then Theorem~\ref{thm:sou-rantzer} implies that there is a rank $k$ matrix $X$ such that $\opnorm{(A/s)X-(B/s)} < 1$ if and only if $k \ge \text{sve}(B/s)$, i.e., $\sigma_{k+1}(B/s) < 1$. This argument shows that for any $s > \max(\sigma_{k+1}(B),\opnorm{(I-U\T{U})B})$, there is a rank $k$ matrix $X$ such that $\opnorm{AX-B}< s$. Thus $\opt = \max(\sigma_{k+1}(B),\opnorm{(I-U\T{U})B})$.

It is interesting and perhaps surprising that the above theorem implies we can obtain a solution that has a value {$\max(\sigma_{k+1}(B),\opnorm{(I-U\T{U})B})$}, which is a simple lower bound on the optimum. This shows that if $\opnorm{(I-U\T{U})B} \le \sigma_{k+1}(B)$, there is a rank $k$ matrix in the column span of matrix $A$ that is as good of an approximation to $B$ in spectral norm as $[B]_k$. Also, if $\opnorm{(I-U\T{U})B} \ge \sigma_{k+1}(B)$, then there is a rank-$k$ matrix in the column space of $A$ that is as good of an approximation to $B$ in spectral norm as $AA^+ B = U\T{U}B$, the projection of $B$ onto the column span of $A$. 

We thus have the following corollary summarizing the discussion above. The corollary was also observed in \cite[see][Section 4]{nambirajan2015topics} in terms of a different parameter they call the critical rank.
\begin{corollary}
	Given matrices $A \in \R^{n \times c}, B \in \R^{n \times d}$ and a parameter $k$,
	\begin{equation*}
		\inf_{\text{rank-}k\ X}\opnorm{AX-B} = \max(\opnorm{(I-AA^+)B},\sigma_{k+1}(B)).
	\end{equation*}
	\label{cor:optimum-value}
\end{corollary}
We give a proof of Theorem~\ref{thm:sou-rantzer} for completeness in Appendix~\ref{subsec:proof-of-sou-rantzer}. Our proof is similar to the proof of \citet{sou2012generalized} with some minor changes.
\input{rank_constrained_lra}
\input{adapting-musco-musco}
\input{approx-oracles-final-results}
\section{Experiments and Implementation}
It is evident that our algorithm is faster than the algorithm of \citet{sou2012generalized} for large matrices $A$ and $B$, as their algorithm cannot make use of the sparsity of the matrices, and also has to compute the eigenvalue decomposition of a dense and large $d \times d$ matrix. Let $n = d = 7000, c = 100$, and $k = 30$. We instantiate an $n \times d$ matrix $B$ with $5\%$ of the entries being non-zero, where each non-zero entry is sampled independently from a uniform distribution on $[0,1]$. The $n \times c$ matrix $A$ is obtained by taking the first $c$ columns of the matrix $B$. With $\epsilon = 0.05$, our algorithm runs in less than $20$ seconds, whereas an implementation of Sou and Rantzer's algorithm runs in around $10$ minutes. For larger values of $n$ and $d$, our algorithm is faster by an even larger factor. An implementation of our algorithm and the above example is available \href{https://gitlab.com/praneeth10/operator-norm-reduced-rank-regression}{here} \footnote{\url{https://gitlab.com/praneeth10/operator-norm-reduced-rank-regression}}.

\section*{Acknowledgments}
The authors would like to thank support from the National Institute of Health (NIH) grant 5R01 HG 10798-2, Office of Naval Research (ONR) grant N00014-18-1-256, and a Simons Investigator Award.

%\printbibliography
\bibliography{main}
\appendix
\input{appendix}
\end{document}

%% file: rank_constrained_lra.tex
\section{Reduced-Rank Regression in Operator Norm}\label{sec:reduced-rank-regression-intro}
We first consider the case when $c,d$ are small. In this case, we could assume that we can compute matrices $U$ and $\Delta$, where $U$ is an orthonormal basis for the column span of matrix $A$, and the matrix $\Delta = \T{B}(I-U\T{U})B$. We give a simple algorithm that demonstrates our techniques. We then extend these ideas to the case when $c,d$ are large, for which computing an orthonormal basis for $A$ and computing $\Delta$ is prohibitively expensive.

From Corollary~\ref{cor:optimum-value}, we have that $\opt = \max(\opnorm{(I-U\T{U})B},\sigma_{k+1}(B))$. Let $\beta$ be such that $(1+\varepsilon)\opt \le \beta \le (1+2\varepsilon)\opt$, which can be found using the Block Krylov algorithm. Throughout the paper we assume we know the value $\beta$.
\begin{lemma}
If there exists a rank-$k$ matrix $X$ such that $\opnorm{UX - B} < \beta$, then
$
    \sigma_{k+1}(\T{U}B(\beta^2I - \Delta)^{-1/2}) < 1.
$
\label{lma:attainable-condition}
\end{lemma}
The proof of this lemma is in Appendix~\ref{subsec:lma:attainable-condition}. The proof of the above lemma also shows that if we can find a matrix $Y$ of rank $k$ such that $\opnorm{Y - \T{U}B(\beta^2I-\Delta)^{-1/2}} \le 1$, then we can obtain a matrix $X = Y(\beta^2I - \Delta)^{1/2}$ such that $\opnorm{UX-B} < \beta$. Thus, we can compute the SVD of the matrix $\T{U}B(\beta^2I - \Delta)^{-1/2}$ and obtain $[\T{U}B(\beta^2I - \Delta)^{-1/2}]_k$ and obtain a solution $[\T{U}B(\beta^2I - \Delta)^{-1/2}]_k(\beta^2I - \Delta)^{1/2}$ of cost $\beta$.

Computing an exact SVD, as required in the proof of above Lemma, is much slower than computing a rank $k$ matrix that satisfies the guarantees of the best rank $k$ matrices approximately. The following lemma shows that we can obtain a solution of cost close to $\beta$ even if we can compute a rank $k$ matrix $Y$ such that $\opnorm{Y - \T{U}B(\beta^2I - \Delta)^{-1/2}} \le 1 + \varepsilon$.
\begin{lemma}
	If $Y$ is a rank $k$ matrix such that $\opnorm{Y - \T{U}B(\beta^2 I - \Delta)^{-1/2}} \le 1 + \varepsilon$, then we obtain that $\opnorm{UY(\beta^2I - \Delta)^{1/2} - B} \le (1+\epsilon)\beta.$ Furthermore,
$
		\opnorm{UY(UY)^{+}B - B} \le (1+\epsilon)\beta.
$
\label{lma:approximation-to-approximation}
\end{lemma}
The proof of this lemma is in Appendix~\ref{subsec:lma:approximation-to-approximation}. The above lemma states that a $1+\varepsilon$ approximation to the best rank-$k$ approximation of the matrix $\T{U}B(\beta^2I - \Delta)^{-1/2}$ in operator norm is sufficient to find a solution of cost $(1 + \epsilon)\beta$ to the reduced-rank regression problem. We can use the Block Krylov algorithm to compute such an approximation. The Block Krylov algorithm of \citet{muscomusco} only needs an oracle to compute matrix-vector products. In the case when $c,d$ are small, we can compute the matrices $U,(\beta^2I - \Delta)^{-1/2}$ and then given arbitrary vectors $v,v'$ we can compute $\T{U}B(\beta^2I - \Delta)^{-1/2}v$ and $(\beta^2I - \Delta)^{-1/2}\T{B}{U}v'$ and hence run the Block Krylov Algorithm. This gives a $1+O(\varepsilon)$ approximation to the reduced-rank regression problem.

When $r,d$ are large, it is expensive to compute the matrices $U,\Delta$ and $(\beta^2I - \Delta)^{-1/2}$. As the analysis of \citet{muscomusco} works only when exact matrix-vector products can be computed, we cannot run the Block Krylov algorithm unless we compute the matrices $U,\Delta$ or at least are able to compute exact matrix vector products with the matrix $\T{U}B(\beta^2I - \Delta)^{-1/2}$. So we analyze their algorithm and show that it  works even using approximate matrix products instead of exact matrix products, given that the error is low enough.

%% file: adapting-musco-musco.tex
\section{Block Krylov Iteration with Approximate Multiplication Oracle}\label{sec:block-krylov-iteration}
\IncMargin{1em}
\begin{algorithm2e}[t]
    \KwIn{$M \in \Real^{n \times d}, k \in \mathbb{Z}, \varepsilon > 0, \text{Oracle}_M : \R^d \times \varepsilon \rightarrow \R^n, \text{Oracle}_{\T{M}} : \R^n \times \varepsilon \rightarrow \R^d$ }
    \KwOut{$Z \in \R^{n \times k}$}
    \DontPrintSemicolon
    % \LineComment{This procedure returns $s$ such that $\opt \le s \le \sqrt{3d_B}\opt$}  
	$G\sim \mathcal{N}(0,1)^{d \times k}$, $\kappa \gets \sigma_1(M)/\sigma_{k+1}(M)$, $q \gets O(({1}/{\sqrt{\varepsilon}})\log(d/\varepsilon))$\;
	$\varepsilon_{\circ} \gets O\left({\varepsilon}/({\kappa^{2+5q}k^{7}C^q}\right))$, $\varepsilon_{\bullet} \gets O\left({\varepsilon^2}/{(48\kappa(\kappa^2(\sqrt{qk})k))}\right)$\;
	\tcc{Let $\circ$ and $\bullet$ denote approximate matrix-vector products using the Oracles with accuracy $\epsilon_\circ$ and $\epsilon_\bullet$, respectively}
% 	\LineComment{Let $ \circ, \bullet $ denote approximate matrix-vector products using $\text{Oracle}_M$ and $\text{Oracle}_{\T{M}}$ with precisions $\varepsilon_{\circ}, \varepsilon_{\bullet}$ respectively}
	$K' \gets [(M\T{M})^{\circ (q-1)/2}M \circ G, (M\T{M})^{\circ (q-3)/2}M \circ G, \ldots, M \circ G]$\;
	$Q' \gets$ Orthonormal basis for $K'$\;
	$[\bar{U},\bar{\Sigma}^2,\T{\bar{U}}] \gets \text{SVD}(\T{Q'}(M \bullet (\T{M}\bullet Q')))$\;
	$\bar{U}_k \gets $First $k$ columns of $\bar{U}$\;
	$Z \gets Q'\bar{U}_k$\;
    \caption{Low Rank Approximation with Approximate Matrix Multiplication}
    \label{alg:musco-musco-adaptation}
\end{algorithm2e}
\DecMargin{1em}

Given a parameter $k$ and an oracle to approximately compute $Mv$ and $\T{M}v'$, given arbitrary vectors $v$ and $v'$, we would like to compute a matrix $Z$ with $k$ orthonormal columns such that
\begin{equation}
	\opnorm{M - Z\T{Z}M} \le (1+\varepsilon)\sigma_{k+1}(M).
	\label{eqn:lra-guarantee}
\end{equation}
Specifically, suppose we have an oracle that, given an arbitrary vector $v$ and approximation parameter $\epsilon_{\circ}$, can compute in time $T(\epsilon_{\circ})$ a vector $M \circ v$ such that
$
	\opnorm{Mv - (M \circ v)} \le \epsilon_{\circ}\opnorm{M}\opnorm{v},
$
and also given an arbitrary vector $v'$ and accuracy parameter $\epsilon_{\circ}$ can compute in time $T(\epsilon_{\circ})$ a vector ${\T{M} \circ v'}$ such that
$
	\opnorm{\T{M}v' - \T{M} \circ v'} \le \epsilon_{\circ}\opnorm{M}\opnorm{v'}
$. We are also given $\kappa = \sigma_1(M)/\sigma_{k+1}(M)$ and we want to compute a matrix $Z$ as in \eqref{eqn:lra-guarantee}.

Our algorithm to compute such a matrix $Z$ is Algorithm~\ref{alg:musco-musco-adaptation}. It is essentially the same as the Block Krylov algorithm of \cite{muscomusco} with exact matrix-vector multiplication replaced by approximate  matrix-vector multiplication with accuracy parameters as defined in our algorithm. 
Our main result for this section is the following theorem that states that the Block Krylov algorithm of \cite{muscomusco} works even with approximate matrix-vector products.
\begin{theorem}
	Let $M \in \R^{n \times d}$, $k \le d$ be a rank parameter, and $\varepsilon > 0$ be an accuracy parameter. Let $\kappa = \sigma_1(M)/\sigma_{k+1}(M)$. Given access to an oracle that can in time $T(\epsilon_{\circ})$ compute vectors $M \circ v$ and $\T{M} \circ v'$ such that
	\begin{equation*}
		\opnorm{M \circ v - Mv} \le \epsilon_{\circ}\opnorm{M}\opnorm{v}\quad \text{and}\quad \opnorm{\T{M}\circ v' - \T{M}v'} \le \epsilon_{\circ}\opnorm{M}\opnorm{v'},
	\end{equation*}
	for any vectors $v$ and $v'$, Algorithm~\ref{alg:musco-musco-adaptation} computes a matrix $Z \in \R^{n \times k}$ with $k$ orthonormal columns such that, with probability $\ge 3/5$,
	$
		\opnorm{(I-Z\T{Z})M} \le (1+\varepsilon)\sigma_{k+1}(M).
	$
	The running time is 
	\begin{equation*}
	O\left(T\left(\frac{\varepsilon}{2\kappa^{5q}k^{11}D^q}\right)qk + T\left(\frac{\varepsilon^2}{192\kappa^2(\sqrt{qk})k}\right)qk\right),
	\end{equation*}
	where $q = O\left(({1}/{\sqrt{\varepsilon}})\log(d/\varepsilon)\right)$ and $D$ is an absolute constant. Further, if the approximations $M \circ v$ are spanned by $M$ for all $v$, then the columns of the matrix $Z$ are also spanned by the matrix $M$.
	\label{thm:main-theorem-krylov}
\end{theorem}
\begin{proofsketch}
	The proof of the Block Krylov algorithm of \cite{muscomusco} first shows that there is a polynomial $p(x)$ that has only odd degree monomials such that the $k$-dimensional column space of the matrix $p(M)G$, where $G$ is a Gaussian matrix with $k$ columns, spans a $(1+\varepsilon)$ approximation. As we do not know how to compute this polynomial $p(x)$, the proof shows that the Krylov Space $K$ spans this matrix $p(M)G$ and then shows that the rank $k$ Frobenius norm approximation of the matrix $M$ inside the Krylov subspace $K$ is also a $1+\varepsilon$ spectral norm rank $k$ approximation.
	
	We adapt their proof to the case when we can compute matrix-vector products only approximately. We first show that the approximate Krylov matrix $K'$ computed by Algorithm~\ref{alg:musco-musco-adaptation} is close to the actual Krylov matrix $K$ in Lemma~\ref{lma:krylov-matrix-error}. However, this lemma is not sufficient to directly prove that the rank-$k$ Frobenius norm approximation of $M$ inside of $K'$ is a $1+\varepsilon$ rank-$k$ spectral approximation, since the matrices $K$ and $K'$ can be very poorly conditioned. Therefore, similar to the matrix $p(M)G$ in \cite{muscomusco}, we define a rank-$k$ matrix $\Apx$ (see Equation~\ref{eqn:apx-matrix-definition}) and show that the matrix $\Apx$ is spanned by $K'$. Then we show in Lemma~\ref{lma:spectral-norm-error-p-M-Apx} that the matrix $\Apx$ is close to $p(M)G$. Using an upper bound on the condition number of the matrix $p(M)G$ (see Lemma~\ref{lma:condition-number-of-p-M-G}), we conclude in Equation~\ref{eqn:diff-pM-apx} that the projection matrices onto the column spaces of the matrices $\Apx$ and $p(M)G$ are close to each other. 
	
	Similar to the argument of \cite{muscomusco}, we encounter the issue that this matrix $\Apx$ cannot be computed as we do not know the parameters of the polynomial $p(x)$, but we do have that this matrix $\Apx$ is spanned by the column space of $K'$. Using this fact, we show that an approximate rank $k$ Frobenius norm approximation of $M$ in the column space of $K'$ is also a $1+\varepsilon$ spectral norm rank $k$ approximation for the matrix $M$. We also show that this approximate rank $k$ Frobenius norm approximation can be computed using approximate matrix-vector product oracles.
\end{proofsketch}

%% file: approx-oracles-final-results.tex
\section{Approximate Oracles and Reduced Rank Regression}\label{sec:final-proof}
Lemma~\ref{lma:approximation-to-approximation} shows that if $Y$ is a rank $k$ matrix such that $\opnorm{Y - \T{U}B(\beta^2I - \Delta)^{-1/2}} \le 1+\varepsilon$, then $\opnorm{UY(UY)^+ B - B} \le (1+\epsilon)\beta$. Based on this result, we prove the following lemma which shows that a low rank-approximation of the matrix $AA^+ B(\beta^2I - \Delta)^{-1/2}$ suffices.
\begin{lemma}
	Let $\tilde{Z} \in \R^{n \times k}$ be a matrix with orthonormal columns such that $$\opnorm{AA^+ B(\beta^2I - \Delta)^{-1/2} - \tilde{Z}\T{\tilde{Z}}AA^+ B(\beta^2I - \Delta)^{-1/2}} \le 1+\varepsilon.$$ Then ${\opnorm{(AA^+ \tilde{Z})(AA^+ \tilde{Z})^+ B - B} \le (1+\epsilon)\beta}$.
	\label{lma:apporximate-subspace-lra-rra}
\end{lemma}
The proof of the lemma is in Appendix~\ref{subsec:lma:apporximate-subspace-lra-rra}. Hence, if we can get a good $k$-dimensional space $\tilde{Z}$ for approximating the matrix $AA^+ B(\beta^2I - \Delta)^{-1/2}$, we can then obtain a good $k$ dimensional space for $B$. We first show that we can instead find a low rank approximation for a matrix $AA^+ BM/\beta$, for a suitable matrix $M$, which will also be a good low rank approximation for  $AA^+ B(\beta^2I - \Delta)^{-1/2}$.

\begin{lemma}
	Given that $\beta \ge (1+\varepsilon)\opt$, there exists a polynomial $ {r}(x)$ of degree at most $t = O\left({1}/{\sqrt{\varepsilon}}\log({\kappa}/{\varepsilon})\right)$ such that for $M =  {r}(\Delta/\beta^2)$, if $\tilde{Z}$ is a matrix such that 
	\begin{equation*}
	\opnorm{{AA^+ BM}/{\beta} - \tilde{Z}\T{\tilde{Z}}({AA^+ BM}/{\beta})}\le 1+\varepsilon,
	\end{equation*}
	then $\opnorm{AA^+ B(\beta^2I - \Delta)^{-1/2} - \tilde{Z}\T{\tilde{Z}}AA^+ B(\beta^2I - \Delta)^{-1/2}} \le 1+O(\varepsilon)$. Furthermore, $\| {r}\|_1 = O((1+\sqrt{2})^{O(\sqrt{1/\varepsilon}\log(\kappa/\varepsilon))}\log(\kappa/\varepsilon)/\varepsilon)$, $\opnorm{M} \le 2/\sqrt{\varepsilon}$, and $\sigma_{\min}(M) \ge 1/2$.
	\label{lma:replacing-neg-square-root}
\end{lemma}
The proof of the above lemma is in Appendix~\ref{subsec:lma:replacing-neg-square-root}. From Theorem~\ref{thm:main-theorem-krylov}, to find a $1+\varepsilon$ approximation for rank $k$ spectral norm low rank approximation (LRA) of the matrix $\mathcal{M'}$, we need only a way to compute the products $\mathcal{M'}v$ and $\T{\mathcal{M'}}v'$ for any vectors $v,v'$. As $ {r}(\Delta/\beta^2)$ is a polynomial in the matrix $\Delta/\beta^2$, it is much easier to design approximate multiplication oracles for the matrix $AA^+ BM/\beta$ than for the matrix $AA^+ B(\beta^2I - \Delta)^{-1/2}$. The following lemma shows that we can compute good approximations to the matrix vector products and then compute a $1+\varepsilon$ approximation to the LRA of matrix $\mathcal{M}' = AA^+ B \frac{ {r}(\Delta/\beta^2)}{\beta}$. 
\begin{lemma}
\label{lma:oracle-time-complexity}
	Given arbitrary vectors $v,v'$ and an accuracy parameter $\epsilonsub{f}$, Algorithms~\ref{alg:oracle-M} and \ref{alg:oracle-M-transpose} compute vectors $y,y'$ such that
	$
		\opnorm{\mathcal{M}'v - y } \le \epsilonsub{f}\opnorm{v}\ \text{and}\ \opnorm{\T{\mathcal{M'}}v' - y'} \le \epsilonsub{f}\opnorm{y}
	$
	in time 
\begin{align*}
T(\epsilonsub{f}) &:= O(t \cdot (\nnz{B} + (\nnz{A} + c^2)\log\left({\kappa(B)^2\|r\|_1}/{(\epsilonsub{f}\varepsilon})\right)))\\
&\quad + O((\nnz{A} + c^2)\log({\kappa(B)}/({\epsilonsub{f}\varepsilon})))
\end{align*}
where $t = O(\sqrt{1/\varepsilon}\log(\kappa/\varepsilon)) $ and $\|r\|_1 = (1+\sqrt{2})^{O(1/\sqrt{\varepsilon}\log(\kappa/\varepsilon))}\log(\kappa/\varepsilon)/\varepsilon$.
\end{lemma}
\subsection{Main Theorem}
We finally have our main theorem that shows that Algorithm~\ref{alg:opnorm-regression} outputs a $1+\varepsilon$ approximation in factored form. The proof of the theorem is in Appendix \ref{subsec:thm:final-theorem}.
\begin{theorem}
	Given matrices $A \in \R^{n \times c}$ and $B \in \R^{n \times d}$, a rank parameter $k \le c$ and an accuracy parameter $\varepsilon$, Algorithm~\ref{alg:opnorm-regression} runs in time 
	\begin{equation*}
			O\left(\left(\frac{\nnz{B} \cdot k}{\varepsilon} + \frac{\nnz{A}\cdot k}{\varepsilon^{1.5}} +\frac{c^2k}{\varepsilon^{1.5}}\right)\cdot \textnormal{polylog}(\kappa, \kappa(AA^+ B),d,k,1/\varepsilon) + c^\omega\right), 
	\end{equation*}
	and with probability $4/5$ outputs a matrix $Z$ with $k$ orthonormal columns, for which $\text{colspan}(Z) \subseteq \text{colspan}(A)$, such that
	$
		\opnorm{Z\T{Z}B - B} \le (1+\varepsilon)\opt. 
	$
	It also outputs matrices $X' \in \R^{c \times k}$ and $X'' \in \R^{k \times d}$ such that
	$
		\opnorm{A(X' \cdot X'') - B} = \opnorm{Z\T{Z}B - B} \le (1+\varepsilon)\opt.
	$
	\label{thm:final-theorem}
\end{theorem}
\IncMargin{1em}
\begin{algorithm2e}[t]
    \caption{Operator Norm Regression}
    \label{alg:opnorm-regression}
    \KwIn{$A \in \R^{n \times c}, B \in \R^{n \times d}, k \in \mathbb{Z}, \varepsilon > 0$}
    \KwOut{$X' \in \R^{c \times k}, X'' \in \R^{k \times d}$}
    \DontPrintSemicolon  
	$\beta \gets (1+\varepsilon/2)\max(\sigma_{k+1}(B),\opnorm{(I-AA^+)B})$\;
	$\Delta \gets \T{B}(I-AA^+)B$\tcc*[r]{Not computed explicitly}
	\tcc{Let $r(x)$ be the polynomial given by Lemma~\ref{lma:replacing-neg-square-root}}
	$\mathcal{M}' \gets (AA^+ B/\beta) {r}(\Delta/\beta^2)$\tcc*[r]{Not computed explicitly}
	$Z \gets \text{Algorithm~\ref{alg:musco-musco-adaptation}}(\mathcal{M}',k,\varepsilon/2,\textsc{ApxProduct},\textsc{ApxProdcutTranspose})$\;
	$X' \gets \textsc{HighPrecisionRegression}(A,Z,1/2)$\;
	$X'' \gets \T{Z}\cdot B$\;
\end{algorithm2e}
\DecMargin{1em}
\subsection{Removing \texorpdfstring{$\kappa(AA^+ B)$}{condition number} Dependence}
We observe that we can add a random rank $k+1$ matrix to $B$ to obtain a matrix $\tilde{B}$ for which $\kappa(AA^+ \tilde{B})$  is bounded in terms of $\kappa(B)$. We also show that any arbitrary vector $v$ can be multiplied with the matrix $\tilde{B}$ in time comparable to $\nnz{B}$.
\begin{lemma}
    Given any matrices $A \in \R^{n \times c}$ and $B \in \R^{n \times d}$, if $\text{rank}(A) \ge k+1$, then there exists a matrix $\tilde{B}$ such that if 
    \begin{equation}
        \opnorm{A\tilde{X} - \tilde{B}} \le (1+\varepsilon/2)\min_{\text{rank-}k\ X}\opnorm{AX - \tilde{B}}
    \end{equation}
    for a rank $k$ matrix $\tilde{X}$, then 
    \begin{equation*}
        \opnorm{A\tilde{X} - {B}} \le (1+\varepsilon)\opt.
    \end{equation*}
    Additionally, $\kappa(AA^\dagger \tilde{B}) = \sigma_1(AA^+ \tilde{B})/\sigma_{k+1}(AA^+ \tilde{B}) \le (Cn/\varepsilon)\sigma_1(B)/\sigma_{k+1}(B)$, and given a vector $v$, $\tilde{B}v$ can be computed in $O(\nnz{B} + (n+d)k)$ time.
    \label{lma:removing-dependence}
\end{lemma}
The proof of this lemma is in Appendix~\ref{subsec:lma:removing-dependence}. Therefore we run Algorithm~\ref{alg:opnorm-regression} on matrix $\tilde{B}$ and can compute a $(1+\varepsilon)$-approximate solution to the problem $\min_{\text{rank-}k\ X}\opnorm{AX-B}$ in time 
    \begin{equation}
	O\left(\left(\frac{\nnz{B} \cdot k}{\varepsilon} + \frac{(n+d)k^2}{\varepsilon}+ \frac{\nnz{A}\cdot k}{\varepsilon^{1.5}} +\frac{c^2k}{\varepsilon^{1.5}}\right)\cdot \text{polylog}(\kappa, n,d,k,1/\varepsilon) + c^{\omega}\right).
\end{equation}

%% file: appendix.tex
\section{Omitted Proofs from Section~\ref{sec:related-work}}
\subsection{Proof of Theorem~\ref{thm:sou-rantzer}}\label{subsec:proof-of-sou-rantzer}
\begin{proof}
Without loss of generality, we prove the theorem assuming $A$ has orthonormal columns. Thus $U=A$. Let $X$ be an arbitrary matrix such that $\opnorm{UX-B} < 1$. We will give a series of statements equivalent to $\opnorm{UX-B} < 1$ that prove the theorem. Using the fact that for any matrix $A$, $\opnorm{A} < 1$ if and only if $\T{A}A \prec I$, we obtain the equivalent statement
\begin{equation*}
	\T{(UX-B)}(UX-B) \prec I.
\end{equation*}
Writing $B$ as $U\T{U}B + (I-U\T{U})B$, we get another equivalent statement
\begin{equation*}
\T{(UX-U\T{U}B)}(UX-U\T{U}B) \prec I - \T{B}(I-U\T{U})B = I - \Delta.
\end{equation*}
As the LHS of the above relation is a positive semi-definite matrix, we obtain that $I-\Delta \succ 0$ and hence is invertible. Thus, the above condition can be equivalently written as
\begin{equation*}
	(I-\Delta)^{-1/2}\T{(UX-U\T{U}B)}(UX-U\T{U}B)(I-\Delta)^{-1/2} \prec I.
\end{equation*}
Using the fact that $\Delta$ is symmetric and $\T{U}U = I$, we get that the above condition is the same as
\begin{equation*}
	\opnorm{X(I-\Delta)^{-1/2} - \T{U}B(I-\Delta)^{-1/2}} < 1.
\end{equation*}
Thus, we obtain that in the case that $\opnorm{(I-U\T{U})B} < 1$, for an arbitrary matrix $X$, the condition that $\opnorm{UX-B} < 1$ is equivalent to $\opnorm{X(I-\Delta)^{-1/2} - \T{U}B(I-\Delta)^{-1/2}} < 1$. Let $\hat{B} := \T{U}B(I-\Delta)^{-1/2}$. It is easy to see that $X = [\hat{B}]_{\text{sve}(\hat{B})}(I-\Delta)^{1/2}$ satisfies $\opnorm{UX-B} < 1$ and that any matrix $X$ that satisfies $\opnorm{UX-B} < 1$ must have rank at least $\text{sve}(\hat{B})$. All that remains to show is that $\text{sve}(\hat{B}) = \text{sve}(B)$.  We will show that $k^-(I-\T{\hat{B}}\hat{B}) = k^-(I-\T{B}B)$, which completes the proof:
\begin{align*}
I-\T{\hat{B}}\hat{B} &= I - (I-\Delta)^{-1/2}\T{B}U\T{U}B(I-\Delta)^{-1/2}\\
&= I - (I-\Delta)^{-1/2}(\T{B}B - \Delta)(I-\Delta)^{-1/2}\\
&= I - (I-\Delta)^{-1/2}(\T{B}B - I + I - \Delta)(I-\Delta)^{-1/2}\\
&= I - I + (I-\Delta)^{-1/2}(I -\T{B}B)(I-\Delta)^{-1/2}\\
&=(I-\Delta)^{-1/2}(I -\T{B}B)(I-\Delta)^{-1/2}.
\end{align*}
Thus $k^-(I-\T{\hat{B}}\hat{B}) = k^-((I-\Delta)^{-1/2}(I-\T{B}B)(I-\Delta)^{-1/2})$. By Sylvester's law of inertia \cite[p313]{carrell_groups_2017}, $k^-((I-\Delta)^{-1/2}(I-\T{B}B)(I-\Delta)^{-1/2}) = k^-(I-\T{B}B)$. Therefore
\begin{equation*}
	\text{sve}(\hat{B}) = k^{-}(I-\T{\hat{B}}\hat{B}) = k^{-}(I-\T{B}B) = \text{sve}(B).
\end{equation*}
Thus $\text{sve}(B)$ is the optimum value for $\eqref{eqn:sr-problem}$ if it is feasible.
\end{proof}
\section{Omitted Proofs from Section~\ref{sec:reduced-rank-regression-intro}}
\subsection{Proof of Lemma~\ref{lma:attainable-condition}}\label{subsec:lma:attainable-condition}
\begin{proof}
The proof of this lemma is very similar to the proof of Theorem~\ref{thm:sou-rantzer}. Suppose there exists a rank-$k$ matrix $X$ such that $\opnorm{UX - B} < \beta$. We already have $\beta > \opnorm{(I-U\T{U})B}$. The statement $\opnorm{UX-B} < \beta$ implies that
\begin{align*}
    \T{(UX-B)}(UX-B) \preceq \beta^2I.
\end{align*}
We can write $B = U\T{U}B + (I- U\T{U})B$ and obtain that for any matrix $X$, $\T{(UX-B)}(UX-B) = \T{(UX-U\T{U}B)}(UX-U\T{U}B) + \Delta$, which implies that
\begin{equation*}
    \T{(UX - U\T{U}B)}(UX - U\T{U}B) \preceq \beta^2I - \Delta.
\end{equation*}
As $\opnorm{\Delta} = \opnorm{(I-U\T{U})B}^2 < \beta^2$, $\beta^2I - \Delta$ is invertible, which implies that
\begin{equation*}
    (\beta^2I - \Delta)^{-1/2}\T{(UX - U\T{U}B)}(UX - U\T{U}B)(\beta^2I - \Delta)^{-1/2} \preceq I.
\end{equation*}
Thus, we have  
$
    \opnorm{(UX - U\T{U}B)(\beta^2I - \Delta)^{-1/2}} = \opnorm{X(\beta^2I - \Delta)^{-1/2} - \T{U}B(\beta^2I - \Delta)^{-1/2}}
$
is less than or equal to $1$.
As $X$ is a matrix of rank $k$, the matrix $X(\beta^2I - \Delta)^{-1/2}$ also has rank $k$. Therefore 
\begin{align*}
	\sigma_{k+1}(\T{U}B(\beta^2I - \Delta)^{-1/2}) &= \opnorm{[\T{U}B(\beta^2I - \Delta)^{-1/2}]_k - \T{U}B(\beta^2I - \Delta)^{-1/2}}\\
	&\le \opnorm{X(\beta^2I - \Delta)^{-1/2} - \T{U}B(\beta^2I - \Delta)^{-1/2}}\\
	&\le 1.
\end{align*}
\end{proof}
\subsection{Proof of Lemma~\ref{lma:approximation-to-approximation}}
\label{subsec:lma:approximation-to-approximation}
\begin{proof}
Suppose $Y$ is a rank $k$ matrix such that $\opnorm{Y - \T{U}B(\beta^2I - \Delta)^{-1/2}} \le 1 + \varepsilon$.
Then we have $\opnorm{Y(\beta^2I - \Delta)^{1/2}(\beta^2I - \Delta)^{-1/2} - \T{U}B(\beta^2I - \Delta)^{-1/2}} \le 1 + \varepsilon$ and therefore
\begin{equation*}
    (\beta^2I - \Delta)^{-1/2}\T{(Y(\beta^2 I - \Delta)^{1/2} - \T{U}B)}(Y(\beta^2 I - \Delta)^{1/2} - \T{U}B)(\beta^2I - \Delta)^{-1/2} \preceq (1 + \varepsilon)^2I.
\end{equation*}
Multiplying the above relation on both sides with $(\beta^2I - \Delta)^{1/2}$ on the left and the right, we obtain
\begin{equation*}
    \T{(Y(\beta^2I - \Delta)^{1/2} - \T{U}B)}(Y(\beta^2I - \Delta)^{1/2} - \T{U}B) \preceq (1 + \varepsilon)^2(\beta^2I - \Delta).
\end{equation*}
Using $\T{U}U = I$ and adding $\Delta$ to both sides, we conclude that
\begin{equation*}
    \opnorm{UY(\beta^2I - \Delta)^{1/2} - B} \le \sqrt{\opnorm{(1+\varepsilon)^2\beta^2I}} \le (1+\epsilon)\beta.
\end{equation*}
Now $Y$ is a matrix that has rank at most $k$. We also have $\opnorm{UYZ - B} \ge \opnorm{UY(UY)^{+}B - B}$ for any matrix $Z$. Therefore
$
    \opnorm{UY(UY)^{+}B - B} \le \opnorm{UY(\beta^2I - \Delta)^{1/2} - B} \le (1+\varepsilon)\beta.
$
\end{proof}
\section{Omitted Proofs from Section~\ref{sec:block-krylov-iteration}}
\subsection{Error in Computing Krylov Subspace}
Given a matrix $M \in \R^{n \times d}$, an integer $k \le d$ and an odd integer $q \ge 0$, the Krylov subspace is defined by 
\begin{equation*}
	K = [(M\T{M})^{(q-1)/2}MG,\ (M\T{M})^{(q-3)/2}MG,\ \cdots,\ (M\T{M})^{1}MG,\ MG]
\end{equation*}
where $G$ is a $d \times k$ matrix with i.i.d. normal entries. Using the algorithm to approximately multiply a vector with the matrices $M$ and $\T{M}$, we compute an approximation to the matrix $K$ defined above. For any vector $v$, define $(M\T{M})^{\circ 0}v := v$ and for $i > 0$, define $(M\T{M})^{\circ i}v := M \circ (\T{M} \circ ((M\T{M})^{\circ (i-1)}v))$ (recall $M \circ v$ is the approximation to $Mv$ computed by the oracle). The notation is similarly extended to approximate matrix multiplication using the oracle. Now we define the matrix
\begin{equation*}
K' = [(M\T{M})^{\circ (q-1)/2}M\circ G,\ (M\T{M})^{\circ (q-3)/2}M\circ G,\ \cdots,\ (M\T{M})^{\circ 1}M\circ G,\ M\circ G].
\end{equation*}
Let $Q$, $Q'$ denote orthonormal bases for the matrices $K$ and $K'$ respectively. We now bound $\frnorm{K - K'}$ and the time required to compute $K'$ using the following lemma.
\begin{lemma}
For any matrix $M \in \R^{n \times d}$, matrix $G \in \R^{d \times k}$ and an odd integer $q$, let $\Delta_{i,G} :=(M\T{M})^{(i-1)/2}MG - (M\T{M})^{\circ (i-1)/2}M\circ G$ and matrices $K,K' \in \R^{n \times qk}$ be as defined above. Then
\begin{equation*}
	E_{i,G} := \frnorm{\Delta_{i,G}} \le 8\epsilon_{\circ}(2^{i/2}\opnorm{M}^i\frnorm{G})
\end{equation*}
for $i = 1,3,5,\ldots,q$ and
$\frnorm{K - K'} \le O(\epsilon_{\circ}\frnorm{G}\opnorm{M}^{q+1}2^{(q+1)/2})$. The matrix $K'$ can be computed in $O(T(\epsilon_{\circ})qk)$ time. 
\label{lma:krylov-matrix-error}
\end{lemma}
\begin{proof}
	For an arbitrary vector $v$ and $i$ odd, let $\Delta_i := {(M\T{M})^{(i-1)/2}Mv - (M\T{M})^{\circ (i-1)/2}M\circ v}$. Let $E_i = \opnorm{\Delta_i}$. We have $E_1 = \opnorm{\Delta_1} = \opnorm{Mv - M \circ v} \le \opnorm{M}\opnorm{v}$. We now define a recurrence relation between $E_{i}$ and $E_{i-2}$ and then bound $E_i$ using this recurrence. We have
	\begin{align*}
		\Delta_i &= (M\T{M})^{(i-1)/2}Mv - (M\T{M})^{\circ (i-1)/2}M \circ v\\
		&= (M\T{M})(M\T{M})^{(i-3)/2}Mv - (M\T{M})^{\circ 1}(M\T{M})^{\circ (i-3)/2}M \circ v\\
		&= (M\T{M})[(M\T{M})^{(i-3)/2}Mv -(M\T{M})^{\circ (i-3)/2}M \circ v]\\
		&\quad + [(M\T{M})^{1}(M\T{M})^{\circ (i-3)/2}M \circ v - (M\T{M})^{\circ 1}(M\T{M})^{\circ (i-3)/2}M \circ v]\\
		&= (M\T{M})\Delta_{i-2} + [(M\T{M})^{1}(M\T{M})^{\circ (i-3)/2}M \circ v - (M\T{M})^{\circ 1}(M\T{M})^{\circ (i-3)/2}M \circ v].
	\end{align*}
	Therefore, by the triangle inequality of $\opnorm{\cdot}$,
	\begin{align*}
		E_i &\le \opnorm{M\T{M}\Delta_{i-2}} + \opnorm{(M\T{M})^{1}(M\T{M})^{\circ (i-3)/2}M \circ v - (M\T{M})^{\circ 1}(M\T{M})^{\circ (i-3)/2}M \circ v}\\
		&\le \opnorm{M}^2E_{i-2}+ \opnorm{(M\T{M})^{1}(M\T{M})^{\circ (i-3)/2}M \circ v - (M\T{M})^{\circ 1}(M\T{M})^{\circ (i-3)/2}M \circ v}.
	\end{align*}
	Let $v' := (M\T{M})^{\circ (i-3)/2}M \circ v$. We now bound $\opnorm{M\T{M}v' - (M\T{M})^{\circ 1}v'}$:
	\begin{align*}
		\opnorm{M\T{M}v' - (M\T{M})^{\circ 1}v'} &= \opnorm{M\T{M}v' - M \circ (M \circ v')}\\
		&\le \opnorm{M\T{M}v' - M(\T{M} \circ v')} + \opnorm{M(\T{M} \circ v') - M \circ (\T{M} \circ v')}\\
		&\le \opnorm{M}\opnorm{\T{M}v' - \T{M} \circ v'} + \epsilon_{\circ}\opnorm{M}\opnorm{\T{M} \circ v'}\\
		&\le \epsilon_{\circ}\opnorm{M}^2\opnorm{v'} + \epsilon_{\circ}\opnorm{M}(\epsilon_{\circ}\opnorm{M}\opnorm{v'} + \opnorm{\T{M}v'})\\
		&\le 3\epsilon_{\circ}\opnorm{M}^2\opnorm{v'}.
	\end{align*}
	As $v' = (M\T{M})^{(i-3)/2}Mv - \Delta_{i-2}$, we get $\opnorm{v'} \le \opnorm{(M\T{M})^{(i-3)/2}Mv} + \opnorm{\Delta_{i-2}} \le \opnorm{M}^{i-2}\opnorm{v} + E_{i-2}$. Therefore we finally obtain that
	\begin{align*}
		E_i \le \opnorm{M}^2E_{i-2} + 3\epsilon_{\circ}\opnorm{M}^2\opnorm{v'} \le \opnorm{M}^2E_{i-2} + 3\epsilon_{\circ}\opnorm{M}^2(\opnorm{M}^{i-2}\opnorm{v} + E_{i-2})\\
		 \le (1+3\epsilon_{\circ})\opnorm{M}^2E_{i-2} + 3\epsilon_{\circ}\opnorm{M}^i\opnorm{v}.
	\end{align*}
	Solving this recurrence relation we obtain that
	\begin{align*}
		E_i &\le (1+3\epsilon_{\circ})^{(i-1)/2}\opnorm{M}^{i-1}E_1 + (1+(1+3\epsilon_{\circ}) + \cdots + (1+3\epsilon_{\circ})^{(i-3)/2})(3\epsilon_{\circ}\opnorm{M}^i\opnorm{v})\\
		&\le \epsilon_{\circ}(1+2^{(i-1)/2}(3\epsilon_{\circ}))\opnorm{M}^i\opnorm{v} + 2^{(i-1)/2}(3\epsilon_{\circ})\opnorm{M}^i\opnorm{v}\\
		&\le 8(\epsilon_{\circ}2^{i/2}\opnorm{M}^i\opnorm{v}).
	\end{align*}
	In the above inequalities, we used the standard inequality $(1+x)^n \le 1+2^nx$ if $0 \le x \le 1$. Thus for any arbitrary vector $v, \opnorm{(M\T{M})^{\circ (i-1)/2}M \circ v - (M\T{M})^{(i-1)/2}Mv} \le 8\epsilon_{\circ}2^{i/2}\opnorm{M}^i\opnorm{v}$ and therefore for the Gaussian matrix $G$,
	\begin{equation*}
		E_{i,G} = \frnorm{(M\T{M})^{\circ (i-1)/2}M \circ G - (M\T{M})^{(i-1)/2}MG} \le 8\epsilon_{\circ}2^{i/2}\opnorm{M}^i\frnorm{G}.
	\end{equation*}
	We then have that $\frnorm{K - K'} \le O(\epsilon_{\circ}\frnorm{G}\opnorm{M}^{q+1}2^{(q+1)/2})$. In computing the matrix $K'$ we make $O(qk)$ calls to each of the oracles and therefore take $O(T(\epsilon_{\circ})qk)$ time.
\end{proof}

\citet{muscomusco} consider a polynomial $p(x)$ such that the column space of the matrix $p(M)G$ is spanned by $K$. They then argue that the column span of $p(M)G$ is a ``good'' $k$-dimensional subspace to project $M$ onto and then conclude that the best rank $k$ approximation of $M$ inside the span of $K$ satisfies \eqref{eqn:lra-guarantee}. Although we have an upper bound on $\frnorm{K - K'}$ from the above lemma, we cannot directly argue that the best rank $k$ approximation of $M$ inside $K'$ satisfies the guarantee of \eqref{eqn:lra-guarantee}, as the matrix $K$ might be very poorly conditioned. 

To overcome this issue, we first show that the matrix $p(M)G$ has a bounded condition number with $O(1)$ probability and that $K'$ spans a matrix $\Apx$ that is close to $p(M)G$. We then show that the span of the matrix $\Apx$ is a good subspace to project the matrix $M$ onto and then conclude that the best rank $k$ approximation of $M$ inside the span of $K'$ satisfies \eqref{eqn:lra-guarantee}.

\subsection{Condition Number of the matrix \texorpdfstring{$p(M)G$}{p(M)G} and existence of good rank \texorpdfstring{$k$}{k} subspace inside an approximate Krylov Subspace}
Throughout this section let $\alpha = \sigma_{k+1}(M)$ and $\gamma = \varepsilon/2$. Let $q$ be an odd integer and $T(x)$ be the degree $q$ Chebyshev polynomial. Define 
\begin{equation}
 {p}(x) := (1+\gamma)\alpha \frac{T(x/\alpha)}{T(1+\gamma)}.	
\label{eqn:definition-p-x}
\end{equation}
The following lemma bounds $\sigma_1( {p}(M))/\sigma_{k+1}( {p}(M))$ which lets us bound $\kappa( {p}(M)G)$.
\begin{lemma}
	\label{lma:condition-number-of-p-M}
	If $M \in \R^{n \times d}$ is a matrix such that $\sigma_1(M)/\sigma_{k+1}(M) = \kappa$, then $$\sigma_1(p(M))/\sigma_{k+1}(p(M)) \le (3\kappa)^q.$$
\end{lemma} 
First, we have the following lemma that shows that $T(x) \ge 1$ for all $x \ge 1$ for the Chebyshev Polynomial $T$ of any degree $d$.
\begin{lemma}
	If $T_d(x)$ is the degree $d$ Chebyshev Polynomial, then for all $d \ge 0$ and for all $x \ge 1$, $T_{d+1}(x) \ge T_d(x) \ge 1$. 
	\label{lma:chebychev-lowerbound}
\end{lemma}
\begin{proof}
	We prove the theorem using induction on the degree $d$. We have $T_0(x) = 1$ and $T_1(x) = x$. Thus $T_1(x) \ge T_0(x) \ge 1$ for $x \ge 1$. Assume that for all $d<n$ and $x \ge 1$, $T_{d+1}(x) \ge T_d(x) \ge 1$. If we now prove that $T_{n+1}(x) \ge T_n(x) \ge 1$, we are done by induction.
	
	We have $T_{n+1}(x) = 2xT_n(x) - T_{n-1}(x) = T_n(x) + [T_n(x) - T_{n-1}(x)] + (2x-2)T_n(x)$. As $x\ge 1$ and by the induction hypothesis $T_n(x) \ge T_{n-1}(x) \ge 1$, we obtain that $T_{n+1}(x) \ge T_{n}(x) \ge 1$. Thus for all $d \ge 0$ and $x \ge 1$, $T_{d+1}(x) \ge T_d(x) \ge 1$.
\end{proof}

Recall $ {p}(x) = (1+\gamma)\alpha \frac{T(x/\alpha)}{T(1+\gamma)} = (1+\varepsilon/2)\sigma_{k+1}\frac{T(x/\alpha)}{T(1+\gamma)}$.

\begin{lemma}
	If $x \ge \alpha > 0$, then $p(x) \le (1+\gamma)\alpha\frac{3^q(x/\alpha)^q}{T(1+\gamma)}$.
	\label{lma:upperbound-p}
\end{lemma}
\begin{proof}
By a standard property, the sum of absolute values of coefficients of the degree-$q$ Chebyshev polynomial is bounded above by $3^q$. Thus $T(x/\alpha) = \sum_{i=1}^q T_i(x/\alpha)^i \le \sum_{i=1}^q|T_i|(x/\alpha)^i \le (x/\alpha)^q\sum_{i=1}^q |T_i| \le 3^q(x/\alpha)^q$, where we use the fact that $(x/\alpha) \ge 1$. Thus $p(x) = (1+\gamma)\alpha T(x/\alpha)/T(1+\gamma) \le (1+\gamma)\alpha3^q(x/\alpha)^q/T(1+\gamma)$.
\end{proof}
\begin{proof}[Proof of Lemma~\ref{lma:condition-number-of-p-M}]
	We bound $\sigma_1( {p}(M))$ and $\sigma_{k+1}( {p}(M))$, and then infer an upper bound on $\frac{\sigma_1( {p}(M))}{\sigma_{k+1}( {p}(M))}$. Let $\sigma_1 \ge \sigma_2 \ge \ldots \ge \sigma_d \ge 0$ be the singular values of the matrix $M$. Then we have that $| {p}(\sigma_1)|,| {p}(\sigma_2)|,\ldots,| {p}(\sigma_d)|$ are the singular values of the matrix $ {p}(M)$. Consider any $i \le k+1$. We have $\sigma_{i} \ge \sigma_{k+1} = \alpha$. Therefore,
\begin{equation*}
	 {p}(\sigma_{i}) = (1+\gamma)\sigma_{k+1}\frac{T(\sigma_i/\sigma_{k+1})}{T(1+\gamma)} \ge \frac{(1+\gamma)\sigma_{k+1}}{T(1+\gamma)}.
\end{equation*}
Here we use Lemma~\ref{lma:chebychev-lowerbound} to lower bound the value of $T(\sigma_i/\sigma_{k+1})$ by $1$. Therefore at least $k+1$ singular values of $ {p}(M)$ are at least $\frac{(1+\gamma)\sigma_{k+1}}{T(1+\gamma)}$, which implies $\sigma_{k+1}( {p}(M)) \ge \frac{(1+\gamma)\sigma_{k+1}}{T(1+\gamma)}$.

Now for any $i \le k+1$, $ {p}(\sigma_i) \le (1+\gamma)\sigma_{k+1}(3^q\kappa^q)/T(1+\gamma)$ by Lemma~\ref{lma:upperbound-p}. For any $i \ge k+1$, we have that $\sigma_{i} \le \sigma_{k+1}$ and $| {p}(\sigma_i)| = (1+\gamma)\sigma_{k+1}|T(\sigma_i/\sigma_{k+1})|/T(1+\gamma) \le (1+\gamma)\sigma_{k+1}/T(1+\gamma)$ by a well known property of Chebyshev polynomials that $|T(x)| \le 1$ for all $x \in [-1,1]$. Therefore
\begin{equation}
	\opnorm{p(M)} = \sigma_1(p(M)) = \max_i |p(\sigma_i(M))| \le (1+\gamma)\sigma_{k+1}\frac{3^q\kappa^q}{T(1+\gamma)}.
	\label{eqn:p-M-operator-norm}
\end{equation}
Thus, $\sigma_1(p(M))/\sigma_{k+1}(p(M)) \le 3^q\kappa^q$.
\end{proof}

We now bound the condition number of the matrix $ {p}(M)G$ where $G$ is a Gaussian matrix with $k$ columns. We use results from \cite{rudelson-vershynin-extreme-singular-values} to bound the maximum and minimum singular values of $G$ with $O(1)$ probability and then use the above lemma to obtain bounds on the extreme singular values of $ {p}(M)G$.
\begin{lemma}
	If $G \in \R^{d \times k}$ is a matrix of i.i.d. normal entries and $M \in \R^{n \times d}$ is a matrix such that $\sigma_1(M)/\sigma_{k+1}(M) = \kappa$, then with probability $\ge 4/5$,
	\begin{equation*}
	    \kappa(p(M)G) = \sigma_{\max}(p(M)G)/\sigma_{\min}(p(M)G) \le Ck3^q\kappa^q,
	\end{equation*}
	for an absolute constant $C > 0$.
	\label{lma:condition-number-of-p-M-G}
\end{lemma}
 \begin{lemma}
If $A \in \R^{n \times d}$ is a matrix with $\sigma_1(A)/\sigma_k(A) \le \kappa_1$ and  $G \in \R^{d \times k}$ is a matrix with i.i.d. normal entries, then for $d$ greater than a constant, with probability $\ge 4/5$, the matrix $AG$ has full rank and has $\sigma_1(AG)/\sigma_k(AG) \le Ck(\sigma_1(A)/\sigma_k(A))$, where $C > 0$ is an absolute constant.
\end{lemma}
\begin{proof}
	Let $A = U\Sigma\T{V}$ be the singular value decomposition of $A$ with $U \in \R^{n \times n},\Sigma \in \R^{n \times d}$ and $\T{V} \in \R^{d \times d}$. Let $G' = \T{V}G$. As rows of $\T{V}$ are orthonormal and entries of $G$ are i.i.d. normal random variables, we obtain that $G'$ is a also a matrix of i.i.d. normal random variables of size $d \times k$. Let $\Sigma_k \in \R^{k \times d}$ be the first $k$ rows of $\Sigma$. For any vector $x$,
	\begin{equation*}
		\opnorm{AGx} = \opnorm{U\Sigma\T{V}Gx} = \opnorm{\Sigma G'x} \ge \opnorm{\Sigma_kG'x}.
	\end{equation*}
	Thus, $\min_{x : \opnorm{x} = 1}\opnorm{AGx} \ge \min_{x:\opnorm{x} = 1}\opnorm{\Sigma_kG'x}$. We have that
	\begin{equation*}
		\Pr[\sigma_{\min}(G') \le \frac{1}{20C}(\sqrt{d} - \sqrt{k-1})] \le \left(\frac{1}{20}\right)^{d-k+1} + e^{-cd}
	\end{equation*}
	for some absolute constants $c$ and $C$ by Theorem~1.1 of \cite{gaussian-smallest-singular-value}. Thus for large enough $d$, with probability $\ge 9/10$, we have
	\begin{equation*}
		\sigma_{\min}(G') \ge \frac{1}{20C}(\sqrt{d} - \sqrt{k-1}).
	\end{equation*}
	Thus $\min_{x:\opnorm{x}=1}\opnorm{\Sigma_k G'x} \ge \sigma_{\min}(\Sigma_k)\sigma_{\min}(G') \ge \frac{\sigma_{k}(A)}{20C}(\sqrt{d} - \sqrt{k-1})$. Similarly, for large enough $d$, we have with probability $\ge 9/10$ that $\sigma_{\max}(G') \le D(\sqrt{d} + \sqrt{k})$ for an absolute constant $D$ by Proposition~2.4 of \cite{rudelson-vershynin-extreme-singular-values} and therefore $\max_{x:\opnorm{x}=1}\opnorm{AGx} = \max_{x:\opnorm{x}=1}\opnorm{\Sigma G'x} \le D\sigma_1(A)(\sqrt{d} + \sqrt{k})$. Therefore with probability $\ge 4/5$,
	\begin{equation*}
		\kappa(AG) = \frac{\sigma_{\max}(AG)}{\sigma_{\min}(AG)} \le 20CD\frac{\sigma_1(A)}{\sigma_k(A)}\frac{\sqrt{d} + \sqrt{k}}{\sqrt{d} - \sqrt{k-1}}.
	\end{equation*}
	The maximum of this expression occurs at $d = k$ and is at most $4k$. Therefore with probability $\ge 4/5$, for $d$ at least some constant, $\kappa(AG) \le 40CDk(\sigma_1(A)/\sigma_k(A))$.
\end{proof}
\begin{proof}[Proof of Lemma~\ref{lma:condition-number-of-p-M-G}]
Using the above lemma, we have that with probability $\ge 4/5$,
\begin{equation*}
	\kappa( {p}(M)G) = \frac{\sigma_{\max}( {p}(M)G)}{\sigma_{\min}( {p}(M)G)} \le Ck\frac{\sigma_1( {p}(M))}{\sigma_k( {p}(M))} \le Ck\frac{\sigma_1( {p}(M))}{\sigma_{k+1}( {p}(M))}
\end{equation*}
for an absolute constant $C$. The last inequality follows from $\sigma_{k}( {p}(M)) \ge \sigma_{k+1}( {p}(M))$. From Lemma~\ref{lma:condition-number-of-p-M}, we have $\frac{\sigma_{1}( {p}(M))}{\sigma_{k+1}( {p}(M))} \le 3^q\kappa^q$. Therefore
	$\kappa( {p}(M)G) \le Ck3^q\kappa^q$ with probability $\ge 4/5$.
\end{proof}

 The bound on the condition number of $ {p}(M)G$ enables us to conclude that if the Frobenius norm error between $ {p}(M)G$ and a matrix $\Apx$ is small, then the projection matrices onto the column spaces of the matrices $ {p}(M)G$ and $\Apx$ are close. Specifically, we use the following lemma.
\begin{lemma}
	Let $A$ and $B$ be full column rank matrices such that $\opnorm{A - B} \le \delta\opnorm{A}$. Let $\kappa(A)$ denote the condition number of the matrix $A$ i.e., $\kappa(A) = \sigma_{\max}(A)/\sigma_{\min}(A)$. Let $U$ and $V$ denote an orthonormal basis for matrices $A$ and $B$, respectively.
	If $\delta \le 1/(2\kappa(A)) \le 1$, then 
	$
		\opnorm{AA^+ - BB^+} = \opnorm{U\T{U} - V\T{V}} \le  20\delta\kappa(A)^4.
	$
	\label{lma:closeness-projection-matrices}
\end{lemma}
\begin{proof}
	As $A$ and $B$ are full rank matrices, we have $A^+ = (\T{A}A)^{-1}\T{A}$ and $B^+ = (\T{B}B)^{-1}\T{B}$. Let $A - B = \Delta$. We have $\opnorm{\Delta} \le \delta\opnorm{A}$. We first have
	\begin{align*}
		\opnorm{AA^+ - BB^+} &= \opnorm{AA^+ - (A-\Delta)B^+}\\
		&\le \opnorm{A}\opnorm{A^+ - B^+} + \opnorm{\Delta}\opnorm{B^+}\\
		&\le \opnorm{A}\opnorm{A^+ - B^+} + \frac{\opnorm{\Delta}}{\sigma_{\min}(B)}.
	\end{align*}
	Note that $\T{A}A = \T{(B+\Delta)}(B+\Delta) = \T{B}B + \T{\Delta}B + \T{B}\Delta + \T{\Delta}\Delta$. Now,
\begin{align*}
	\opnorm{A^+ - B^+} &= \opnorm{(\T{A}A)^{-1}\T{A} - (\T{B}B)^{-1}\T{B}}\\
	&= \opnorm{(\T{A}A)^{-1}\T{A} - (\T{B}B)^{-1}(\T{A} - \T{\Delta})}\\
	&\le \opnorm{(\T{A}A)^{-1} - (\T{B}B)^{-1}}\opnorm{A} + \opnorm{(\T{B}B)^{-1}}\opnorm{\Delta}\\
	&\le \opnorm{(\T{A}A)^{-1} - (\T{B}B)^{-1}}\opnorm{A} + \frac{\opnorm{\Delta}}{\sigma_{\min}(B)^2}.
\end{align*}
We finally bound $\opnorm{(\T{A}A)^{-1} - (\T{B}B)^{-1}}$.
\begin{align*}
	\opnorm{(\T{A}A)^{-1} - (\T{B}B)^{-1}} &\le \frac{1}{\sigma_{\min}(\T{A}A)}\opnorm{(\T{A}A)((\T{A}A)^{-1} - (\T{B}B)^{-1})}\\
	&\le \frac{1}{\sigma_{\min}(\T{A}A)}\opnorm{I - (\T{A}A)(\T{B}B)^{-1}}\\
	&\le \frac{1}{\sigma_{\min}(\T{A}A)}\opnorm{I - (\T{B}B + \T{\Delta}B + \T{B}\Delta + \T{\Delta}\Delta)(\T{B}B)^{-1}}\\
	&\le \frac{1}{\sigma_{\min}(\T{A}A)}\opnorm{I - I - (\T{\Delta}B + \T{B}\Delta + \T{\Delta}\Delta)(\T{B}B)^{-1}}\\
	&\le \frac{2\opnorm{\Delta}\opnorm{B} + \opnorm{\Delta}^2}{\sigma_{\min}(\T{A}A)\sigma_{\min}(\T{B}B)}.
\end{align*}
We therefore obtain
\begin{align*}
	\opnorm{AA^+ - BB^+} &\le \opnorm{A}^2\opnorm{(\T{A}A)^{-1} - (\T{B}B)^{-1}} + \frac{\opnorm{\Delta}\opnorm{A}}{\sigma_{\min}(\T{B}B)} + \frac{\opnorm{\Delta}}{\sigma_{\min}(B)}\\
	&\le \frac{\opnorm{A}^2}{\sigma_{\min}(\T{A}A)}\frac{2\opnorm{\Delta}\opnorm{B} + \opnorm{\Delta}^2}{\sigma_{\min}(\T{B}B)} + \frac{\opnorm{\Delta}\opnorm{A}}{\sigma_{\min}(\T{B}B)}+ \frac{\opnorm{\Delta}}{\sigma_{\min}(B)}.
\end{align*}
As $\opnorm{A - B} \le \delta\opnorm{A}$, we get that $(1-\delta)\opnorm{A} \le \opnorm{B} \le (1+\delta)\opnorm{A}$. We also have that $\sigma_{\min}(B) \ge \sigma_{\min}(A) - \opnorm{A-B} \ge \opnorm{A}/\kappa(A) - \delta\opnorm{A} \ge \opnorm{A}/2\kappa(A) = \sigma_{\min}(A)/2$ if $\delta < 1/2\kappa(A)$. We can therefore conclude that $\opnorm{AA^+ - BB^+} \le 20\delta\kappa(A)^4$.
\end{proof}
The condition that $\delta$ must be less than $1/2\kappa(A)$ in the above lemma makes sense as otherwise $20\delta\kappa(A)^4 \ge 10\kappa(A)^3 \ge 10$, which is a trivial upper bound on the norm.

Now we construct a matrix $\Apx$ that has its columns spanned by $K'$ and is close to the matrix $ {p}(M)G$. Using the bound on the condition number of the matrix $ {p}(M)G$, we can conclude that the projection matrices onto the column spans of $\Apx$ and $ {p}(M)G$, respectively, are close.

Recall $p(x)$ from \eqref{eqn:definition-p-x}. For $q$ odd, the Chebyshev polynomial of degree $q$ contains only odd degree monomials. So we have
$
	T(x) = T_q x^q + T_{q-2}x^{q-2} + \ldots + T_1x
$
and therefore, the polynomial
$
	 {p}(x) = \frac{(1+\gamma)\alpha}{T(1+\gamma)}\left(\frac{T_q}{\alpha^q}x^q + \frac{T_{q-2}}{\alpha^{q-2}}x^{q-2} +\cdots + \frac{T_1}{\alpha_1}x\right),
$
which implies
\begin{equation*}
	 {p}(M)G = \frac{(1+\gamma)\alpha}{T(1+\gamma)}\left(\frac{T_q}{\alpha^q}(M\T{M})^{(q-1)/2}MG +\cdots + \frac{T_1}{\alpha_1}MG\right).
\end{equation*}
We now define
\begin{equation}
	\Apx = \frac{(1+\gamma)\alpha}{T(1+\gamma)}\left(\frac{T_q}{\alpha^q}(M\T{M})^{\circ(q-1)/2}M\circ G + \cdots + \frac{T_1}{\alpha_1}M\circ G\right).
	\label{eqn:apx-matrix-definition}
\end{equation}
Clearly, the matrix $\Apx$ is spanned by the columns of the matrix $K'$. Using Lemma~\ref{lma:krylov-matrix-error} and properties of Gaussian matrices, the following lemma bounds $\opnorm{\Apx - p(M)G}$.
\begin{lemma}
	For the matrices $p(M)G$ and $\Apx$ defined above, we have with probability $\ge 3/5$
	\begin{equation*}
		\opnorm{ {p}(M)G - \Apx} \le \frnorm{ {p}(M)G - \Apx} \le 64C\epsilon_{\circ}k^{3/2}(3\sqrt{2}\kappa)^q\opnorm{ {p}(M)G}.
	\end{equation*}
	\label{lma:spectral-norm-error-p-M-Apx}
\end{lemma}
 \begin{proof}
	By the triangle inequality,
\begin{align*}
	&\frnorm{ {p}(M)G - \Apx}\\
	&\le \frac{(1+\gamma)\alpha}{T(1+\gamma)}\sum_{\text{odd $i \le q$}} \frac{|T_i|}{\alpha^i}\frnorm{(M\T{M})^{(i-1)/2}MG - (M\T{M})^{\circ (i-1)/2}M\circ G}\\
	&\le \frac{(1+\gamma)\alpha}{T(1+\gamma)}\sum_{\text{odd $i \le q$}} \frac{|T_i|}{\alpha^i}E_{i,G} \\\
	&\le  \frac{(1+\gamma)\alpha}{T(1+\gamma)}\sum_{\text{odd $i \le q$}} \frac{|T_i|}{\alpha^i}8\epsilon_{\circ}(2^{i/2}\opnorm{M}^i\frnorm{G}) && \text{(Lemma~\ref{lma:krylov-matrix-error})}\\
	&\le \frac{(1+\gamma)\sigma_{k+1}(M)}{T(1+\gamma)}8\epsilon_{\circ}\frnorm{G}\sum_{\text{odd $i \le q$}}|T_i|(\sqrt{2}\kappa)^i && \text{($\alpha = \sigma_{k+1}(M)$)}\\
	&\le \frac{(1+\gamma)\sigma_{k+1}(M)}{T(1+\gamma)}8\epsilon_{\circ}\frnorm{G}(3\sqrt{2}\kappa)^q && \text{($\sum_i|T_i| \le 3^q$)}\\
	&\le \opnorm{ {p}(M)}8\epsilon_{\circ}\frnorm{G}(3\sqrt{2}\kappa)^q. && (\text{Equation~\ref{eqn:p-M-operator-norm}})
\end{align*}
We also condition on the following events both of which hold simultaneously with probability $\ge 4/5$.
\begin{itemize}
	\item $\frnorm{G} \le 4\sqrt{dk}$, and
	\item $\opnorm{ {p}(M)G} \ge (1/C)\opnorm{ {p}(M)}(\sqrt{d} - \sqrt{k-1}) \ge (1/2C)\opnorm{p(M)}\sqrt{d}$.
\end{itemize}
Thus, with probability $\ge 4/5$, if $d \ge 4k$,
\begin{equation*}
	\frnorm{ {p}(M)G - \Apx} \le \opnorm{p(M)} (32\epsilon_{\circ})\sqrt{dk}(3\sqrt{2}\kappa)^q \le 64C\epsilon_{\circ}\sqrt{k}(3\sqrt{2}\kappa)^q\opnorm{ {p}(M)G}. 
\end{equation*}
If $k \le d \le 4k$, then $\opnorm{ {p}(M)G} \ge (1/C)\opnorm{ {p}(M)}(\sqrt{d} - \sqrt{k-1}) \ge (1/2C)\opnorm{ {p}(M)}(1/\sqrt{k})$ and $\frnorm{ {p}(M)G - \Apx} \le 64C\epsilon_{\circ}k^{3/2}(3\sqrt{2}\kappa)^q\opnorm{ {p}(M)G}$.
\end{proof}
 Let $Y_1 \in \R^{n \times k}$ be an orthonormal basis for the column span of $ {p}(M)G$ and $Y \in \R^{n \times k}$ be an orthonormal basis for the matrix $\Apx$. We now have from Lemmas~\ref{lma:closeness-projection-matrices} and \ref{lma:spectral-norm-error-p-M-Apx} that
\begin{align*}
	\opnorm{Y\T{Y} - Y_1\T{Y_1}} &\le O(\epsilon_{\circ}k^{3/2}(3\sqrt{2}\kappa)^q\kappa(p(M)G)^4) = O(\epsilon_{\circ}{k^{3/2}}(3\sqrt{2}\kappa)^q(k^43^{4q}\kappa^{4q}))\\
	&= \epsilon_{\circ}C^qk^6\kappa^{5q}
\end{align*}
for some constant $C$. Let $\delta := \epsilon_{\circ}C^qk^6\kappa^{5q}$. Hence 
\begin{equation}
\opnorm{Y_1\T{Y_1} - Y\T{Y}} \le \delta.
\label{eqn:diff-pM-apx}
\end{equation}
For $l \le k$ such that $\sigma_l(M) \ge (1+\varepsilon)\sigma_{k+1}(M)$, let $\mathcal{E}_l = \frnorm{[M]_l}^2 - \frnorm{Y_1\T{Y_1}[M]_l}^2$ and  $\mathcal{E}_{l}' = \frnorm{[M]_l}^2 - \frnorm{Y\T{Y}[M]_l}^2$. \citet[Equation 7]{muscomusco} show that \[\mathcal{E}_l = \frnorm{[M]_l}^2 - \frnorm{Y_1\T{Y_1}[M]_l}^2 \le (\varepsilon/2)\sigma_{k+1}(M)^2.\] Bounding $\mathcal{E}_l$ is one of the important steps in the analysis of \cite{muscomusco}. We obtain a similar bound on $\mathcal{E}_l'$. We further show that if $M_{K',l}$ is the best rank $l$ Frobenius norm approximation of $M$ in colspan$(K')$, then $\frnorm{[M]_l}^2 - \frnorm{M_{K',l}}^2 \le (3\varepsilon/4)\sigma_{k+1}(M)^2$, showing that there is a very good rank-$l$ approximation for $M$ in $\text{colspan}(K')$. We have the following lemma.
\begin{lemma}
	Given a matrix $A$ and a parameter $k$, let $Y_1$ be an orthonormal basis for a $k$ dimensional subspace such that
	$
		\mathcal{E}_l = \frnorm{[M]_l}^2 - \frnorm{Y_1\T{Y_1}[M]_l}^2 \le (\varepsilon/2)\sigma_{k+1}^2
	$
	for all $l \le k$ satisfying $\sigma_l(M) \ge (1+\varepsilon)\sigma_{k+1}(M)$. If $Y$ is an orthonormal basis for another $k$ dimensional subspace for which $\opnorm{Y\T{Y} - Y_1\T{Y_1}} \le \varepsilon/(16\kappa^2\sqrt{k})$, where $\kappa = \sigma_1(M)/\sigma_{k+1}(M)$, then for all such $l$,
	\begin{equation*}
		\mathcal{E}_l' = \frnorm{[M]_l}^2 - \frnorm{Y\T{Y}[M]_l}^2 \le (3\varepsilon/4)\sigma_{k+1}^2.
	\end{equation*}
	There also exists a matrix $Y^l$ with $l$ orthonormal columns with $\text{colspan}(Y^l) \subseteq \text{colspan}(K')$ such that
	$
		\frnorm{[M]_l}^2 -\frnorm{Y^l\T{(Y^l)}M}^2 \le (3\varepsilon/4)\sigma_{k+1}^2.
	$
	\label{lma:existence-of-good-space-inside-k-prime}
\end{lemma}
\begin{proof}
	For any $1 > \epsilonsub{s} > 0$
\begin{align*}
	\frnorm{Y_1\T{Y_1}M_l}^2 &\le (1+\epsilonsub{s})\frnorm{Y\T{Y}M_l}^2 + (1+\frac{1}{\epsilonsub{s}})\frnorm{(Y\T{Y} - Y_1\T{Y_1})M_l}^2\\
	&\le (1+\epsilonsub{s})\frnorm{Y\T{Y}M_l}^2 + (2/\epsilonsub{s})2k\delta^2\sigma_1(M)^2.
\end{align*}
The last inequality follows from the fact that $Y\T{Y} - Y_1\T{Y_1}$ has rank at most $2k$. Therefore
\begin{equation*}
	\frnorm{Y\T{Y}M_l}^2 \ge \frac{1}{1+\epsilonsub{s}}\frnorm{Y_1\T{Y_1}M_l}^2 - \frac{4k\sigma_1(M)^2}{\epsilonsub{s}}\delta^2
\end{equation*}
which implies that
\begin{align*}
	\mathcal{E}_l' &= \frnorm{M_l}^2 - \frnorm{Y\T{Y}M_l}^2\\
	&\le \frnorm{M_l}^2 - \frac{1}{1+\epsilonsub{s}}\frnorm{Y_1\T{Y_1}M_l}^2 + \frac{4k\sigma_1(M)^2}{\epsilonsub{s}}\delta^2\\
	&\le \frac{1}{1+\epsilonsub{s}}(\frnorm{M_l}^2 - \frnorm{Y_1\T{Y_1}M_l}^2) + \epsilonsub{s}\frnorm{M_l}^2 + \frac{4k\sigma_1(M)^2}{\epsilonsub{s}}\delta^2\\
	&\le \frac{1}{1+\epsilonsub{s}}\frac{\varepsilon}{2}\sigma_{k+1}(M)^2 + \epsilonsub{s}k\sigma_1(M)^2+ \frac{4k\sigma_1(M)^2}{\epsilonsub{s}}\delta^2.
\end{align*}
Picking $\epsilonsub{s} = \varepsilon/(8k\kappa^2)$ and if $\delta \le \varepsilon/(16\kappa^2\sqrt{k})$, we obtain that
\begin{equation*}
	\mathcal{E}_l' = \frnorm{M_l}^2 - \frnorm{Y\T{Y}M_l}^2 \le \frac{3\varepsilon}{4}\sigma_{k+1}^2.
\end{equation*}
Recall here that $\kappa = \sigma_1(M)/\sigma_{k+1}(M)$. The matrix $Y\T{Y}M_l$ is a rank $l$ approximation for matrix $M$ inside the column span of $Y$ and hence in the column span of $K'$. Let $Y^l$ be a rank $l$ matrix that forms a basis for the best rank $l$ approximation of $M$ inside the column space of $K'$ i.e.,
\begin{equation*}
	\min_{\text{rank-}l\ B:\text{colspan}(B) \subseteq \text{colspan($K'$)}} \frnorm{M-B}^2 = \frnorm{M - Y^lY^lM}^2.
\end{equation*}
From Lemma~\ref{lma:frobenius-norm-rrr}, note that if $\bar{U}\bar{\Sigma}^2\T{\bar{V}}$ is the singular value decomposition of the matrix $\T{Q'}M\T{M}Q'$ (recall $Q'$ denotes an orthonormal basis for the matrix $K'$), then $Y^l = Q'\bar{U}_l$ where $\bar{U}_l$ denotes the first $l$ columns of the matrix $\bar{U}$.
By the optimality of $Y^l$, $\frnorm{M - Y^l\T{(Y^l)}M}^2 \le \frnorm{M - Y\T{Y}M_l}^2$ which implies that $\frnorm{Y\T{Y}M_l}^2 \le \frnorm{Y^l\T{(Y^l)}M}^2$. Thus $\frnorm{M_l}^2 - \frnorm{Y^l\T{(Y^l)}M}^2 \le \frnorm{M_l}^2 - \frnorm{Y\T{Y}M_l}^2 = \mathcal{E}_l' \le (3\varepsilon/4)\sigma_{k+1}^2$.
\end{proof}
The proof also shows that if $\bar{U}\bar{\Sigma}^2\T{\bar{U}}$ is the singular value decomposition of the positive semi-definite matrix $\T{Q'}M\T{M}Q'$, then $Y^l = Q'\bar{U}_l$ where $\bar{U}_l$ denotes the matrix that contains the first $l$ columns of $\bar{U}$. Let $m \le k$ be the largest integer for which $\sigma_{m}(M) \ge (1+\varepsilon)\sigma_{k+1}(M)$. From the above lemma, the matrix $Y^m$ satisfies $\frnorm{M_l}^2 -\frnorm{Y^m\T{(Y^m)}M}^2 \le (3\varepsilon/4)\sigma_{k+1}(M)^2$. We later show that this implies $\opnorm{M - Y^m\T{(Y^m)}M} \le (1+3\varepsilon/2)\sigma_{k+1}(M)$. Unfortunately, we cannot compute the matrix $\T{Q'}M\T{M}Q'$ exactly as we only have access to an oracle that computes vector products with matrices $M,\T{M}$ approximately. Nevertheless we show that we can compute a matrix $\hat{Y}^m$ based on an approximation to the matrix $\T{Q'}M\T{M}Q'$ and it still satisfies the desired guarantees approximately.

First we have the following lemma that shows if a subspace $Y^m$ is a good approximation for Frobenius norm low rank approximation of $M$ in $m$ dimensions, then the subspace $Y^m$ is also a good subspace for spectral norm rank-$k$ approximation of matrix $M$. It also shows that even if $\hat{Y}^m$ only approximately satisfies the properties of $Y^m$, the matrix $\hat{Y}^m$ spans a good low rank approximation for $M$.
\begin{lemma}
	Given an arbitrary matrix $M$, if an orthonormal basis $Y^{m}$ to an $m$-dimensional subspace, where $m \le k$ is the largest integer such that $\sigma_m(M) \ge (1+\varepsilon)\sigma_{k+1}(M)$,  satisfies
	\begin{equation*}
		\frnorm{M_m}^2 - \frnorm{Y^m\T{(Y^m)}M}^2 \le \varepsilon\sigma_{k+1}(M)^2,
	\end{equation*}
	then $\opnorm{M - Y^m\T{(Y^m)}M} \le (1+2\varepsilon)\sigma_{k+1}(M)$. Additionally if $\hat{Y}^m$ is a matrix with $m$ orthonormal columns such that 
	\begin{equation*}
		\frnorm{M - \hat{Y}^m\T{(\hat{Y}^m)}M}^2 \le \frnorm{M-Y^m\T{(Y^m)}M}^2 + \delta, 
	\end{equation*}
	then $\opnorm{M - \hat{Y}^m\T{(\hat{Y}^m)}M} \le (1+2\varepsilon)\sigma_{k+1}(M)+\sqrt{\delta}$.
	\label{lma:approximation-with-approximate-subspace}
\end{lemma}
\begin{proof}
	As $\frnorm{M_m}^2 - \frnorm{Y^{m}\T{(Y^m)}M}^2 = \frnorm{M}^2 - \frnorm{M-M_m}^2 - \frnorm{Y^{m}\T{(Y^{m})}M}^2 = \frnorm{M-Y^{m}\T{(Y^{m})}M}^2 - \frnorm{M-M_m}^2$, we obtain that
	\begin{equation*}
		\frnorm{M-Y^{m}\T{(Y^m)}M}^2 \le \frnorm{M-M_m}^2 + \varepsilon\sigma_{k+1}(M)^2.
	\end{equation*}
	As an additive error in Frobenius norm translates to additive error in spectral norm for the above case (see Theorem~3.2 from \cite{gu-error-translation-fro-to-op}), we obtain
	\begin{align*}
		\opnorm{M-Y^{m}\T{(Y^m)}M}^2 &\le \opnorm{M-M_m}^2 + \varepsilon\sigma_{k+1}(M)^2 \le \sigma_{m+1}(M)^2 + \varepsilon\sigma_{k+1}(M)^2\\
		&\le (1+4\varepsilon)\sigma_{k+1}(M)^2.
	\end{align*}
	Thus $\opnorm{M - Y^m\T{(Y^m)}M} \le (1+2\varepsilon)\sigma_{k+1}(M)$. Similarly, we have that
	\begin{equation*}
		\frnorm{M - \hat{Y}^m\T{(\hat{Y}^m)}M}^2 \le \frnorm{M - M_m}^2 + \varepsilon\sigma_{k+1}(M)^2 + \delta
	\end{equation*}
	which implies that
	\begin{equation*}
		\opnorm{M - \hat{Y}^m\T{(\hat{Y}^m)}M}^2 \le \opnorm{M - M_m}^2 + \varepsilon\sigma_{k+1}(M)^2 + \delta \le (1+4\varepsilon)\sigma_{k+1}(M)^2 + \delta
	\end{equation*}
	which shows $\opnorm{M - \hat{Y}^m\T{(\hat{Y}^m)}M} \le (1+2\varepsilon)\sigma_{k+1}(M) + \sqrt{\delta}$. 
\end{proof}
The above lemma shows that we need only compute a matrix $\hat{Y}^m$ such that $\frnorm{M - \hat{Y}^m\T{(\hat{Y}^m)}M}^2 \approx \frnorm{M-Y^m\T{(Y^m)}M}^2$. 

We show that using an approximation to matrix $\T{Q'}M\T{M}\T{Q'}$ we can compute such a matrix $\hat{Y}^m$ which shows that $\opnorm{M - \hat{Y}^m\T{(\hat{Y}^m)}M} \le (1+O(\varepsilon))\sigma_{k+1}(M)$. As the value of $m \le k$ is not known, we further show that we can compute a matrix $\hat{Y}^k$ with $k$ orthonormal columns such that $\text{colspan}(M) \supseteq \text{colspan}(K') \supseteq \text{colspan}(\hat{Y}^k) \supseteq \text{colspan}(\hat{Y}^m)$. Therefore we can conclude that $\opnorm{M - \hat{Y}^k\T{(\hat{Y}^k)}M} \le \opnorm{M - \hat{Y}^k\T{(\hat{Y}^k)}M} \le (1+O(\varepsilon))\sigma_{k+1}(M)$. We thus have our final result for low rank approximation.
\subsection{Proof of Theorem~\ref{thm:main-theorem-krylov}}\label{subsec:thm:main-thm-krylov}
\paragraph{Computing top \texorpdfstring{$k$}{k} singular vectors of the matrix \texorpdfstring{$\T{Q'}M\T{M}Q'$}{Q' transpose x M x M transpose x Q'}}
We now show that if $\hat{Y}^m$ are the top $m$ singular vectors of the matrix $\T{Q'}((M\T{M}) \circ Q')$, then 
\begin{equation*}
	\frnorm{M -\hat{Y}^m\T{(\hat{Y}^m)}M}^2 \approx \frnorm{M - Y^m\T{(Y^m)}M}^2.
\end{equation*}
\begin{lemma}
	If $Z_m$ are the top $m$ orthonormal eigenvectors of the matrix $M\T{M}$, then for any matrix $Y$ with $m$ orthonormal columns, 
	\begin{equation*}
		\textnormal{tr}(\T{Z_m}M\T{M}Z_m) \ge \textnormal{tr}(\T{Y}M\T{M}Y).
	\end{equation*}
\end{lemma}
\begin{proof}
	We have $\text{tr}(\T{Z_m}M\T{M}Z_m) = \frnorm{Z_m\T{Z_m}M}^2$ and $\text{tr}(\T{Y}M\T{M}Y) = \frnorm{Y\T{Y}M}^2$. We are given that $Z_m$ are the top $m$ eigenvectors of the matrix $M\T{M}$ and therefore $Z_m$ are the top $m$ singular vectors of the matrix $M$. Therefore for any matrix $Y$ with $m$ orthonormal columns, we have that $\frnorm{Z_m\T{Z_m}M}^2 \ge \frnorm{Y\T{Y}M}^2$ and therefore that $\text{tr}(\T{Z_m}M\T{M}Z_m) \ge \text{tr}(\T{Y}M\T{M}Y)$.
\end{proof}
\begin{lemma}
	Let $M$ be a matrix and $Q$ be an orthonormal basis for an arbitrary $r$ dimensional space. Let $B$ be a positive semi-definite matrix such that $B - \T{Q}M\T{M}Q = \Delta$. Let $Z$ be a matrix whose columns are the top $k$ eigenvectors of the matrix $B$. Then if $Z_m$ denotes the matrix with first $m$ columns of $Z$ for $m=1,\ldots,k$ we have
	\begin{equation*}
		\frnorm{M -(QZ_m)\T{(QZ_m)}M}^2 \le \frnorm{M - Q(\T{Q}M)_m}^2 + 2m\frnorm{\Delta}.
	\end{equation*}
	\label{lma:singular-vectors-with-approximate-matrix}
\end{lemma}
\begin{proof}
	Let $Z^*$ be the matrix whose columns are the top $k$ eigenvectors of the matrix $\T{Q}M\T{M}Q$ and $Z^*_m$ be the first $m$ columns of $Z^*$. Thus $Q(\T{Q}M)_m = Q(Z^*_m\T{(Z^*_m)}\T{Q}M) = (QZ_{m}^*)\T{(QZ_m^*)}M$. Now,
\begin{align*}
	\frnorm{(QZ_m)\T{(QZ_m)}M}^2 &= \frnorm{\T{(QZ_m)}M}^2\\
	&= \text{tr}(\T{Z_m}\T{Q}M\T{M}QZ_m)\\
	&= \text{tr}(\T{Z_m}(\T{Q}M\T{M}Q + \Delta)Z_m) - \text{tr}(\T{Z_m}\Delta Z_m)\\
	&= \text{tr}(\T{Z_m}B Z_m) - \text{tr}(\T{Z_m}\Delta Z_m)\\
	&\ge \text{tr}(\T{(Z_m^*)}BZ_m^*) - m\frnorm{\Delta}\\
	&\quad \text{(Since $\text{tr}(\T{Z_m}\Delta Z_m) = \text{tr}(\Delta Z_m\T{Z_m}) \le \frnorm{\Delta}\frnorm{Z_m\T{Z_m}} \le \frnorm{\Delta} \cdot m$)}\\
	&= \text{tr}(\T{(Z_m^*)}(\T{Q}M\T{M}Q)Z_m^*) - \text{tr}(\T{(Z_m^*)}\Delta Z_m^*) - m\frnorm{\Delta}\\
	&= \text{tr}(QZ_m^*\T{(Z_m^*)}\T{Q}M\T{M}QZ_m^*\T{(Z_m^*)}\T{Q}) - \text{tr}(\T{(Z_m^*)}\Delta Z_m^*) - m\frnorm{\Delta}\\
	&\ge \frnorm{(QZ_m^*)\T{(QZ_m^*)}M}^2 - 2m\frnorm{\Delta}.
\end{align*}
Thus,
\begin{equation*}
	\frnorm{M-(QZ_m)\T{(QZ_m)}M}^2 \le \frnorm{M - (QZ_m^*)\T{(QZ_m^*)}M}^2 + 2m\frnorm{\Delta}, 
\end{equation*}
which concludes the proof.
\end{proof}
Hence if $\tilde{\Apx}$ is a positive semi-definite matrix such that $\frnorm{\tilde{\Apx} - \T{Q'}M\T{M}Q'}$ is small and if $Z_m$ denotes the top $m$ singular vectors of the matrix $\tilde{\Apx}$, we can conclude by Lemma~\ref{lma:approximation-with-approximate-subspace} that $\opnorm{M - (Q'Z_m)\T{(Q'Z_m)}M}$ is close to $\sigma_{k+1}(M)$.

We now show that we can compute such a matrix $\tilde{\Apx}$. Let $\Xi = \T{Q'}((M\T{M}) \circ Q')$ (recall that $\circ$ denotes matrix multiplication using the noisy oracle). Let $\tilde{Apx} = \text{psd}((\Xi + \T{\Xi})/2)$. Then the following lemma shows that $\tilde{Apx}$ is close to $\T{Q'}M\T{M}Q'$.
\begin{lemma}
Given matrices $M \in \R^{n \times d}$ and $Q' \in \R^{n \times t}$ where $Q'$ is a matrix with $t$ orthonormal columns, if for all vectors $v,v'$, $\opnorm{M \circ v - Mv} \le \epsilon_{\circ}\opnorm{M}\opnorm{v}$ and $\opnorm{\T{M} \circ v' - \T{M}v'} \le \epsilon_{\circ}\opnorm{M}\opnorm{v'}$, and $\Xi := \T{Q'}(M\T{M}) \circ Q'$, then 
\begin{equation*}
	\frnorm{\textnormal{psd}((\Xi + \T{\Xi})/2) - \T{Q'}M\T{M}Q'} \le (6\epsilon_{\circ}\opnorm{M}^2)\sqrt{t}.
\end{equation*}
Let $\tilde{Apx} = \textnormal{psd}((\Xi + \T{\Xi})/2)$. The matrix $\tilde{\Apx}$ can be computed in time $O(2tT(\epsilon_{\circ}) + t^3)$.
\label{lma:approximating-qtmmtq}
\end{lemma}
\begin{proof}
	Let $k_i$ be the $i^{\text{th}}$ column of the matrix $K'$ and $E_i = \opnorm{\T{Q'}(M\T{M})\circ k_i - \T{Q'}(M\T{M})k_i}$. Then
	\begin{align*}
		E_i &= \opnorm{\T{Q'}(M\T{M})\circ k_i - \T{Q'}(M\T{M})k_i}\\
		&\le \opnorm{(M\T{M})\circ k_i - (M\T{M})k_i}\\
		&= \opnorm{M \circ (\T{M} \circ k_i) - M(\T{M}k_i)}\\
		&\le \opnorm{M \circ (\T{M} \circ k_i) - M(\T{M} \circ k_i) + M(\T{M} \circ k_i) - M(\T{M}k_i)}\\
		&\le \opnorm{M \circ (\T{M} \circ k_i) - M(\T{M} \circ k_i)} + \opnorm{M(\T{M} \circ k_i) - M(\T{M}k_i)}\\
		&\le \epsilon_{\circ}\opnorm{M}\opnorm{\T{M} \circ k_i} + \opnorm{M}\epsilon_{\circ}\opnorm{M}\opnorm{k_i}\\
		&\le \epsilon_{\circ}\opnorm{M}(\opnorm{\T{M}k_i} + \epsilon_{\circ}\opnorm{M}\opnorm{k_i}) + \opnorm{M}^2\epsilon_{\circ}\opnorm{k_i}\\
		&\le 3\epsilon_{\circ}\opnorm{M}^2. && \text{(Since $\opnorm{k_i} = 1$)}
	\end{align*}
	Thus $\frnorm{\T{Q'}M\T{M}Q' - \Xi}^2 = \sum_{i=1}^t\opnorm{\T{Q'}M\T{M}k_i - \T{Q'}(M\T{M})\circ k_i}^2 \le (3\epsilon_{\circ}\opnorm{M}^2)^2t$ which implies that $\frnorm{\T{Q'}M\T{M}Q' - \Xi} \le (3\epsilon_{\circ}\opnorm{M}^2)\sqrt{t}$. Now as $\T{Q'}M\T{M}Q'$ is a symmetric matrix, $\frnorm{\T{Q'}M\T{M}Q' - (\Xi+\T{\Xi})/2} \le (3\epsilon_{\circ}\opnorm{M}^2)\sqrt{t}$. As $\T{Q'}M\T{M}Q'$ is itself a positive semidefinite matrix,
	\begin{equation*}
		\frnorm{\text{psd}((\Xi + \T{\Xi})/2) - (\Xi + \T{\Xi})/2} \le \frnorm{\T{Q'}M\T{M}Q' - (\Xi+\T{\Xi})/2} \le (3\epsilon_{\circ}\opnorm{M}^2)\sqrt{t}.
	\end{equation*}
	Finally, by the triangle inequality we obtain that $\frnorm{\T{Q'}M\T{M}Q' - \tilde{Apx}} = \frnorm{\T{Q'}M\T{M}Q' - \text{psd}((\Xi + \T{\Xi})/2)} \le 6\epsilon_{\circ}\opnorm{M}^2\sqrt{t}$. The time required to compute matrix $\Xi$ is $2tT(\epsilon_{\circ}) + nt^2$ and $\text{psd}((\Xi + \T{\Xi})/2)$ is $O(t^3)$. Thus, the matrix $\tilde{\Apx}$ can be computed in time $O(2tT(\epsilon_{\circ}) + t^3)$.
\end{proof}
\begin{proof}[Proof of Theorem~\ref{thm:main-theorem-krylov}]
Let $q = O((1/\sqrt{\varepsilon})\log(d/\varepsilon))$. Algorithm~\ref{alg:musco-musco-adaptation} computes the Krylov subspace $K'$ with
\begin{equation*}
	\epsilon_{\circ} = \frac{\varepsilon}{16\kappa^{2+5q}k^{7}C^q}
\end{equation*}
for an absolute constant $C$. Let $Y_1$ be an orthonormal basis for $p(M)G$ and $Y$ be an orthonormal basis for the matrix $\Apx$ (defined in \eqref{eqn:apx-matrix-definition}). Then by \eqref{eqn:diff-pM-apx} we have that $\opnorm{Y\T{Y} - Y_1\T{Y_1}} \le \varepsilon/(16\kappa^2\sqrt{k})$. If $m \le k$ is the largest integer such that $\sigma_{m}(M) \ge (1+\varepsilon)\sigma_{k+1}(M)$, by Lemma~\ref{lma:existence-of-good-space-inside-k-prime}, there exists a $d$ dimensional subspace $Y^m$ inside the column span of $K'$ such that
\begin{equation*}
	\frnorm{M_m}^2 -\frnorm{Y^m\T{(Y^m)}A}^2 \le (3\varepsilon/4)\sigma_{k+1}^2.
\end{equation*}
If $\Xi$ is now computed with $\varepsilon_o = \varepsilon^2/(48\kappa^2(\sqrt{qk})k)$, then by Lemma~\ref{lma:approximating-qtmmtq},
\begin{equation*}
	\frnorm{\T{Q'}M\T{M}Q' - \tilde{\Apx}} \le \frac{\varepsilon}{8k}\sigma_{k+1}^2.
\end{equation*}
Now if $Z_k$ denotes the first $k$ singular vectors of the matrix $\tilde{Apx}$, and $Z_m$ denotes the first $m$ columns of $Z_k$, then by Lemma~\ref{lma:singular-vectors-with-approximate-matrix}, we get that
\begin{align*}
	\frnorm{M - (Q'Z_m)\T{(Q'Z_m)}M}^2 &\le \frnorm{M - Q'(\T{Q'}M)_m}^2 + 2m(\frac{\varepsilon^2}{8k}\sigma_{k+1}^2)\\
	&\le \frnorm{M - Q'(\T{Q'}M)_m}^2 + \frac{\varepsilon^2}{4}\sigma_{k+1}^2.
\end{align*}
Finally, by Lemma~\ref{lma:approximation-with-approximate-subspace}, we obtain that
\begin{equation*}
	\opnorm{M - (Q'Z_m)\T{(Q'Z_m)}M} \le (1+3\varepsilon/2)\sigma_{k+1} + \sqrt{(\varepsilon^2/4)\sigma_{k+1}^2} \le (1+2\varepsilon)\sigma_{k+1}.
\end{equation*}
Also $\opnorm{M - (Q'Z_k)\T{(Q'Z_k)}M} \le \opnorm{M - (Q'Z_m)\T{(Q'Z_m)}M} \le (1+2\varepsilon)\sigma_{k+1}(M)$ since $Q'Z_k$ has orthonormal columns and $\text{colspan}(Q'Z_k) \supseteq \text{colspan}(Q'Z_m)$.
Thus in time
\begin{equation*}
	T\left(\frac{\varepsilon}{\kappa^{5q}k^{7}C^q}\right)qk + T\left(\frac{\varepsilon^2}{48\kappa^2(\sqrt{qk})k}\right)qk, 
\end{equation*}
Algorithm~\ref{alg:musco-musco-adaptation} computes a $1+2\varepsilon$ approximation. Scaling the value of $\varepsilon$ gives us the result. If the approximations $M \circ v$ are spanned by the column space of $M$ for all vectors $v$, then the columns of $K'$ are spanned by the matrix $M$. Thus the columns of $Q'$ are also spanned by $M$,  which implies that the columns of the matrix $Q'Z_m$ are spanned by $M$.
\end{proof}
\section{Omitted Proofs in Section~\ref{sec:final-proof}}
\subsection{Proof of Lemma~\ref{lma:apporximate-subspace-lra-rra}}\label{subsec:lma:apporximate-subspace-lra-rra}
\begin{proof}
	Define $Z := U^T\tilde{Z}$. We have
\begin{align*}
	1+\varepsilon &\ge \opnorm{AA^+ B(\beta^2I - \Delta)^{-1/2} - \tilde{Z}\T{\tilde{Z}}AA^+ B(\beta^2I - \Delta)^{-1/2}}\\
	&\ge \opnorm{U\T{U}B(\beta^2I - \Delta)^{-1/2} - \tilde{Z}\T{\tilde{Z}}U\T{U}B(\beta^2I - \Delta)^{-1/2}}\\
	&\ge \opnorm{U\T{U}B(\beta^2I - \Delta)^{-1/2} - U\T{U}\tilde{Z}\T{\tilde{Z}}U\T{U}B(\beta^2I - \Delta)^{-1/2}}\\
	&= \opnorm{U\T{U}B(\beta^2I - \Delta)^{-1/2} - UZ\T{Z}\T{U}B(\beta^2I - \Delta)^{-1/2}}\\
	&= \opnorm{\T{U}B(\beta^2I - \Delta)^{-1/2} - Z\T{Z}\T{U}B(\beta^2I - \Delta)^{-1/2}}
\end{align*}
which implies using Lemma~\ref{lma:approximation-to-approximation} that $UZ = U\T{U}\tilde{Z} = AA^+ \tilde{Z}$ is a good space to project the columns of $B$ onto, i.e.,
\begin{equation*}
	\opnorm{(AA^+ \tilde{Z})(AA^+ \tilde{Z})^{+}B - B} \le (1+\epsilon)\beta.
\end{equation*}
\end{proof}
\subsection{Proof of Lemma~\ref{lma:replacing-neg-square-root}}\label{subsec:lma:replacing-neg-square-root}
\paragraph{Polynomial Approximation of \texorpdfstring{$(1-x)^{-1/2}$}{(1-x) raised to -0.5}.}
We want to obtain a polynomial $p(x)$ such that $|p(x) - (1-x)^{-1/2}| \le \delta$ in the interval $x \in [0,1/(1+\varepsilon)]$. Consider the Taylor expansion of $(1-x)^{-1/2}$:
\begin{equation*}
	(1-x)^{-1/2} = \sum_{j=0}^\infty \frac{(2j)!}{2^{2j}j!^2}x^j.
\end{equation*}
The above series converges for all $|x| < 1$. Let $q(x)$ be the Taylor series up to $T$ terms. Then for $1 > x \ge 0$, we have  $0 \le q(x) \le (1-x)^{-1/2}$ and for $0 \le x \le 1/(1+\varepsilon)$
\begin{equation*}
	(1-x)^{-1/2} - q(x) = \sum_{j=T}^\infty \frac{(2j)!}{2^{2j}j!^2}x^j \le \sum_{j=T}^{\infty}x^j = \frac{x^T}{1-x} \le \frac{(1+\varepsilon)}{\varepsilon(1+\varepsilon)^T} = \frac{1}{\varepsilon(1+\varepsilon)^{T-1}}.
\end{equation*}
Thus, if $T-1 \ge 4\log(1/(\varepsilon\delta))/\varepsilon \ge \log(1/\varepsilon\delta)/\log(1+\varepsilon)$, we have $(1+\varepsilon)^{T-1} \ge 1/\varepsilon\delta$ which implies that
\begin{equation*}
	0 \le (1-x)^{-1/2} - q(x) \le \delta
\end{equation*}
for all $0 \le x \le 1/(1+\varepsilon)$. So, there is a degree $t = O(\log(1/\varepsilon\delta)/\varepsilon)$ polynomial that uniformly approximates $(1-x)^{-1/2}$ up to an error $\delta$ in the interval $[0,1/(1+\varepsilon)]$. Now, we further approximate the degree $t$ polynomial $q(x)$ with a degree $\tilde{O}(\sqrt{t})$ polynomial.

First we have the following theorem.
\begin{theorem}[Theorem~3.3 of \cite{sachdeva-vishnoi}]
	For any positive integers $s$ and $d$, there is a degree $d$ polynomial $p_{s,d}(x)$ that satisfies
	\begin{equation*}
		\sup_{x \in [-1,1]}|p_{s,d}(x) - x^s| \le 2e^{-d^2/2s}.
	\end{equation*}
	Further, this polynomial $p_{s,d}$ is defined as follows
	\begin{equation*}
		p_{s,d}(x) = \E_{Y_1,\ldots,Y_s}[T_{|D|}(x)\mathbb{I}[|D| \le d]]
	\end{equation*}
	where $Y_1,\ldots,Y_s$ are independent Rademacher random variables, $D = \sum_{i=1}^s Y_i$ and $\mathbb{I}$ denotes the indicator function.
\end{theorem}
Clearly the polynomial $p_{s,d}$ is defined as a weighted linear combination of Chebyshev polynomials of various degrees at most $d$. With $d = \sqrt{2s\log(1/\delta)}$, we have that
\begin{equation*}
	\sup_{x \in [-1,1]}|p_{s,d}(x) - x^s| \le 2e^{^{-\log(1/\delta)}} \le 2\delta.
\end{equation*}
Thus, given an arbitrary degree $t$ polynomial $q(x) = \sum_{i=0}^t q_ix^i$, where $q_0,\ldots,q_t$ are the coefficients of the polynomial, then the degree $d$ polynomial $ {r}(x) = \sum_{i=0}^t q_i p_{i,d}(x)$ with $d = \sqrt{2t\log(1/\delta)}$ satisfies
\begin{align*}
	\sup_{x \in [-1,1]} |q(x) -  {r}(x)| &= \sup_{x \in [-1,1]}|\sum_{i=0}^tq_ix^i - \sum_{i=0}^tq_ip_{i,d}(x)|\\
	&\le \sup_{x \in [-1,1]}\sum_{i=0}^t|q_i||x^i - p_{i,d}(x)|\\
	&\le \sup_{x \in [-1,1]}\sum_{i=0}^t|q_i|2\delta\\
	&= 2\|q\|_1\delta.
\end{align*}
We now bound $\|r\|_1$. We have 
\begin{align*}
\|r\|_1 = \|\sum_i q_ip_{i,d}(x)\| &\le \sum_i |q_i|\|p_{i,d}(x)\|_1\\
& = \sum_i |q_i|\|\E_{Y_1,\ldots,Y_s}[T_{|D|}(x)\mathbb{I}[|D| \le d]]\|_1\\
&\le \sum_i |q_i|\E_{Y_1,\ldots,Y_s}[\|T_{|D|}(x)\mathbb{I}[|D| \le d]]\|_1]\\
&\le \sum_i |q_i|\frac{1}{2}(1+\sqrt{2})^d = \frac12(1+\sqrt{2})^d\|q\|_1.
\end{align*}
Here we use the fact that $\| \cdot \|_1$ is convex over polynomials and that the sum of absolute values of coefficients of a Chebyshev polynomial of degree $d$ is bounded by $(1+\sqrt{2})^d$. Thus we have the following lemma.
\begin{lemma}
	Given any polynomial $q(x)$ of degree $t$, there exists a polynomial $ {r}(x)$ of degree $d = \sqrt{2t\log(2\|q\|_1/\delta)}$ such that
	\begin{equation*}
		\sup_{x \in [-1,1]} |q(x) -  {r}(x)| \le \delta
	\end{equation*}
	and $\|r\|_1 \le (1+\sqrt{2})^{d}\|q\|_1$.
\end{lemma}

We already saw that the polynomial $q(x) = \sum_{j=0}^t\frac{(2j)!}{2^{2j}j!^2}x^j$ satisfies $|q(x) - (1-x)^{-1/2}| \le \delta$ for $x \in [0,1/(1+\varepsilon)]$ if $t = O(\log(1/\varepsilon\delta)/\varepsilon)$. We also have $\|q\|_1 = \sum_{j=0}^t|(2j)!/(2^{2j}(j!)^2)| \le t+1$. Thus by the above lemma, we can compute a polynomial $ {r}(x)$ of degree $d = O(\sqrt{t\log(t/\delta)}) = O(\frac{1}{\sqrt{\varepsilon}}\log(1/\varepsilon\delta))$ such that
\begin{equation*}
	\sup_{x \in [0,1/(1+\varepsilon)]} | {r}(x) - (1-x)^{-1/2}| \le \sup_{x \in [0,1/(1+\varepsilon)]}|(1-x)^{-1/2} - q(x)| + \sup_{x \in [-1,1]}|q(x) -  {r}(x)| \le 2\delta
\end{equation*}
and we also have $\|r\|_1 =  O((1+\sqrt{2})^dt) = O((1+\sqrt{2})^{O(\sqrt{1/\varepsilon}\log(1/\varepsilon\delta))}\log(1/\varepsilon\delta)/\varepsilon)$. We summarize this in the following lemma.
\begin{lemma}
	Given $\varepsilon, \delta > 0$, there exists a polynomial $ {r}(x)$ of degree $O(\frac{1}{\sqrt{\varepsilon}}\log(1/\varepsilon\delta))$ and $\|r\|_1 =O((1+\sqrt{2})^{O(\sqrt{1/\varepsilon}\log(1/\varepsilon\delta))}\log(1/\varepsilon\delta)/\varepsilon)$ such that
	\begin{equation*}
		\sup_{x \in [0,1/(1+\varepsilon)]}| {r}(x) - (1-x)^{-1/2}| \le \delta.
	\end{equation*}
	\label{lma:poly-approximation-inverse-square-root}
\end{lemma}
\begin{lemma}[Matrix Approximation Lemma]
If $A \in \R^{n \times n}$ is a positive semidefinite matrix with $\lambda_{\max}(A) < 1$ and if $ {r}(x)$ is a polynomial such that 
\begin{equation*}
	\sup_{x \in [0,\lambda_{\max}(A)]}| {r}(x) - (1-x)^{-1/2}| \le \delta, 
\end{equation*}
then $\opnorm{ {r}(A) - (I-A)^{-1/2}} \le \delta$.
\label{lma:matrix-approximation-lemma}
\end{lemma}
\begin{proof}
	Let $A = VD\T{V}$ be the eigenvalue decomposition of $D$ with $D = \text{diag}(\lambda_1,\ldots,\lambda_n)$ where $\lambda_{\max} = \lambda_1\ge \ldots \ge \lambda_n \ge 0$. Then $(I-A)^{-1/2} = V(I-D)^{-1/2}\T{V}$ and $ {r}(A) = V {r}(D)\T{V}$. Therefore
	\begin{align*}
		\opnorm{ {r}(A) - (I-A)^{-1/2}} &= \opnorm{V((I-D)^{-1/2} -  {r}(D))\T{V}}\\
		&= \opnorm{(I-D)^{-1/2} -  {r}(D)}\\
		&= \max_i |(1-\lambda_i)^{-1/2} -  {r}(\lambda_i)|\\
		&\le \sup_{x \in [0,\lambda_{\max}(A)]}|(1 - x)^{-1/2} -  {r}(x)| \le \delta.
	\end{align*}
	Here we use the fact that $0 \le \lambda_1, \ldots, \lambda_n \le \lambda_{\max}(A)$.
\end{proof}
As $\Delta$ is a positive semidefinite matrix such that $\beta^2 \ge (1+\varepsilon)\opnorm{\Delta}$, then $\opnorm{\Delta/\beta^2} \le 1/(1+\varepsilon)$ and hence we can compute a polynomial $ {r}(x)$ of degree $O(\frac{1}{\sqrt{\varepsilon}}\log(1/\varepsilon\delta))$ such that \[\opnorm{ {r}(\Delta/\beta^2) - (I-\Delta/\beta^2)^{-1/2}} \le \delta.\]
\paragraph{Modified Problem.}
Instead of considering the matrix $\mathcal{M} = AA^+ B(\beta^2I - \Delta)^{-1/2}$ for low rank approximation, we consider the matrix $\mathcal{M}' = AA^+ BM/\beta$ for $M= {r}(\Delta/\beta^2)$, where $ {r}(x)$ is a low degree polynomial, and argue that a $(1+\varepsilon)$-approximate LRA solution for the matrix $\mathcal{M'}$ is a $1+2\varepsilon$ approximation for the LRA problem on matrix $\mathcal{M}$.
\begin{proof}[Proof of Lemma~\ref{lma:replacing-neg-square-root}]
Recall  $\Delta = \T{B}(I-AA^+)B$, and therefore $\opnorm{\Delta} = \opnorm{(I-AA^+)B}^2$. Given that $\beta \ge (1+\varepsilon)\max(\opnorm{(I-AA^+)B},\sigma_{k+1}(B))$, we have $\beta^2 \ge (1+\varepsilon)^2\opnorm{\Delta} \ge (1+\varepsilon)\opnorm{\Delta}$. Thus, $\opnorm{\Delta/\beta^2} \le 1/(1+\varepsilon)$. 

As $\opnorm{\Delta/\beta^2} \le 1/(1+\varepsilon)$, we approximate $(I-\Delta/\beta^2)^{-1/2}$ with the matrix $ M = {r}(\Delta/\beta^2)$ where $ {r}(x) = \sum_{i=0}^t r_ix^i$ is a polynomial of degree $t = O(\frac{1}{\sqrt{\varepsilon}}\log(\frac{1}{\varepsilon\delta}))$ given by Lemma~\ref{lma:poly-approximation-inverse-square-root}. By Lemma~\ref{lma:matrix-approximation-lemma}, the matrix $ {r}(\Delta/\beta^2) = \sum_{i=0}^t r_i(\Delta/\beta^2)^i$ satisfies
\begin{align*}
    \opnorm{(I - \Delta/\beta^2)^{-1/2} - M} &= \opnorm{(I - \Delta/\beta^2)^{-1/2} -  {r}(\Delta/\beta^2)}\\
     & = \opnorm{(I-\Delta/\beta^2)^{-1/2} - \sum_{i=0}^t r_i \left(\frac{\Delta}{\beta^2}\right)^i}\le \delta.
\end{align*}
As $\opnorm{\Delta/\beta^2} \le 1/(1+\varepsilon)$ and $\Delta/\beta^2$ is a positive semidefinite matrix, we have $\sigma_{\max}(I- \Delta/\beta^2) \le 1$ and $\sigma_{\min}(I - \Delta/\beta^2) \ge \varepsilon/(1+\varepsilon)$. Therefore $\sigma_{\max}((I-\Delta/\beta^2)^{-1/2}) \le \sqrt{{(1+\varepsilon)}/{\varepsilon}}$ and $\sigma_{\min}((I-\Delta/\beta^2)^{-1/2}) \ge 1$. By Weyl's inequality, we obtain that 
\begin{equation*}
\sigma_{\max}(M) \le \sqrt{(1+\varepsilon)/\varepsilon} + \delta\quad \text{and}\quad \sigma_{\min}(M) \ge 1 - \delta.
\end{equation*}
By sub-multiplicativity of the spectral norm
\begin{align*}
	\opnorm{AA^+ B(\beta^2I - \Delta)^{-1/2}-\frac{AA^+ BM}{\beta}} &\le \frac{\opnorm{AA^+ B}}{\beta}\opnorm{(I-(\Delta/\beta^2))^{-1/2} - M}\\
	&\le \frac{\opnorm{AA^+ B}}{\beta}\delta.
\end{align*}
Using Weyl's inequality, we obtain
\begin{equation}
	\sigma_{k+1}\left(\frac{AA^+ BM}{\beta}\right) \le \sigma_{k+1}(AA^+ B(\beta^2I - \Delta)^{-1/2}) + \frac{\opnorm{AA^+ B}}{\beta}\delta \le 1 + \frac{\opnorm{AA^+ B}}{\beta}\delta.
	\label{eqn:weyls-inequality-conclusion-k-plus-one}
\end{equation}
The last inequality follows as there exists a rank $k$ matrix with $\opnorm{AX-B} \le \beta$. If we can now find a rank $k$ matrix $Z$ with orthonormal columns such that
\begin{equation}
	\opnorm{Z\T{Z}\frac{AA^+ BM}{\beta} - \frac{AA^+ BM}{\beta}} \le (1+\varepsilon)\sigma_{k+1}\left(\frac{AA^+ BM}{\beta}\right),
	\label{eqn:guarantee-we-show}
\end{equation}
then
\begin{align*}
	&\opnorm{Z\T{Z}AA^+ B(\beta^2I - \Delta)^{-1/2} - AA^+ B(\beta^2I - \Delta)^{-1/2}}\\
	&\le \opnorm{Z\T{Z}\frac{AA^+ BM}{\beta} - \frac{AA^+ BM}{\beta}} + \opnorm{(I-Z\T{Z})\left(\frac{AA^+ BM}{\beta} - AA^+ B(\beta^2I - \Delta)^{-1/2}\right)}\\
	&\le (1+\varepsilon)\sigma_{k+1}\left(\frac{AA^+ BM}{\beta}\right) + \frac{\opnorm{AA^+ B}}{\beta}\delta\\
	&\le (1+\varepsilon)(1+2\opnorm{AA^+ B}(\delta/\beta)). 
\end{align*}
The last inequality follows from \eqref{eqn:weyls-inequality-conclusion-k-plus-one}. If $\delta$ is chosen to be less than $\varepsilon/4\kappa$ where $\kappa = \sigma_1(B)/\sigma_{k+1}(B)$, then
\begin{align*}
	&\opnorm{Z\T{Z}AA^+ B(\beta^2I - \Delta)^{-1/2} - AA^+ B(\beta^2I - \Delta)^{-1/2}} \\
	&\le (1+\varepsilon)\left(1+2\opnorm{AA^+ B}(\delta/\beta)\right)\\
	&\le (1+\varepsilon)\left(1 + 2 \frac{\opnorm{AA^+ B}}{\beta}\frac{\varepsilon\sigma_{k+1}(B)}{4\sigma_1(B)}\right)\\
	&\le 1+2\varepsilon
\end{align*}
as $\opnorm{AA^+ B} \le \sigma_1(B)$ and $\beta \ge (1+\varepsilon)\sigma_{k+1}(B)$.
This implies that if $(1-x)^{-1/2}$ is approximated by a polynomial $ {r}(x)$ uniformly in the interval $[0,1/(1+\varepsilon)]$ with an error at most $\varepsilon/4\kappa$, and if matrix $Z$ is an orthonormal basis for a space that spans a $1+\varepsilon$ rank $k$ approximation in spectral norm for the matrix $AA^+ B \frac{ {r}(\Delta/\beta^2)}{\beta}$, then
\begin{equation*}
	\opnorm{AA^+ Z(AA^+ Z)^+ B - B} \le (1+6\varepsilon)\beta = (1+O(\varepsilon))\opt.
\end{equation*}
We obtain the proof by appropriately scaling $\varepsilon$.
\end{proof}
\subsection{Proof of Lemma~\ref{lma:oracle-time-complexity}}
\IncMargin{1em}
\begin{algorithm2e}[t]
    \caption{Oracle$_\mathcal{M'}$}
    \label{alg:oracle-M}
    \KwIn{$v \in \R^d$, $\epsilonsub{r} > 0$}
    \KwOut{$y \in \R^n$}
    \DontPrintSemicolon
 	\tcc{Let $ {r}(x)$ be the polynomial as in Lemma~\ref{lma:replacing-neg-square-root}}
	{$t \gets \text{degree}(r)$}\;
	{$\epsilonsub{reg} \gets O\left({\epsilonsub{r}}/{\kappa\|r\|_1}\right)$}\;
	{$y \gets 0$}\;
	{$\Apx_0 \gets v$}\;
	\For{$i=0,\ldots,t$}{
	{$y \gets y + r_i\Apx_i$}\;
	{$\Apx_{i+1} \gets \T{B}B \cdot\Apx_{i} - \T{B} \cdot (\textsc{HighPrecisionRegression}(A,B\cdot \Apx_i,\epsilonsub{reg}))$}\;
	}
	{$y \gets (\text{HighPrecisionRegression}(A,B\cdot y,\epsilonsub{reg}))/\beta$}\;
\end{algorithm2e}
\DecMargin{1em}
\IncMargin{1em}
\begin{algorithm2e}[t]
    \caption{Oracle$_\mathcal{\T{M'}}$}
    \label{alg:oracle-M-transpose}
    \KwIn{$v \in \R^d$, $\epsilonsub{r} > 0$}
    \KwOut{$y \in \R^n$}
    \DontPrintSemicolon
    % \LineComment{This procedure returns $s$ such that $\opt \le s \le \sqrt{3d_B}\opt$}  
    % \Procedure{ApproximateProductTranspose}{$v,\varepsilon$}
 	\tcc{Let $ {r}(x)$ be the polynomial as in Lemma~\ref{lma:replacing-neg-square-root}}
	{$t \gets \text{degree}(r)$}\;
	{$\epsilonsub{reg} \gets O\left({\epsilonsub{r}}/{\kappa\|r\|_1}\right)$}\;
	{$y \gets 0$}\;
	{$\Apx_0 \gets \T{B} \cdot (\textsc{HighPrecisionRegression}(A,v,\epsilonsub{reg}))$}\;
	\For{$i=0,\ldots,t$}{
	    {$y \gets y + r_i\Apx_i$}\;
	    {$\Apx_{i+1} \gets \T{B}B \cdot \Apx_{i} - \T{B} \cdot (\textsc{HighPrecisionRegression}(A,B\cdot \Apx_i,\epsilonsub{reg}))$}\;
	}
	{$y \gets y/\beta$}\;
\end{algorithm2e}
\DecMargin{1em}
Throughout the analysis, we assume $\opnorm{AA^+ B} \ge \varepsilon\opnorm{B}$. Suppose that $\opnorm{AA^+ B} \le \varepsilon\opnorm{B}$. Let $z$ be the top singular vector of matrix $B$. Then
\begin{align*}
	\opnorm{B}^2 &= \opnorm{Bz}^2\\
	&= \opnorm{AA^+ Bz}^2 + \opnorm{(I-AA^+)Bz}^2\\
	&\le \varepsilon^2\opnorm{B}^2 + \opnorm{(I-AA^+ )Bz}^2.
\end{align*}
Thus, $\opnorm{(I-AA^+)B}^2 \ge \opnorm{(I-AA^+)Bz}^2 \ge (1-\varepsilon^2)\opnorm{B}^2$. Therefore $\opt \ge \sqrt{1-\varepsilon^2}\opnorm{B}$ which implies $\opnorm{B} \le (1/\sqrt{1-\varepsilon^2})\opt \le (1+\varepsilon)\opt$ for $\varepsilon \le 1/2$. Thus $\opnorm{A(0) - B} \le (1+\varepsilon)\opt$ and hence we have a trivial $(1+\varepsilon)$-approximate solution. Thus, we can assume $\opnorm{AA^+ B} \ge \varepsilon\opnorm{B}$.
%
%{\color{red} First try to taylor expand the term $(\beta^2I - \Delta)^{-1/2}$ as before. Then try using the trick of running power method on $AA^+ B(\beta^2I - \Delta)^{-1}\T{B}\T{(AA^{+})} = AA^+ B(\beta^2I - \Delta)^{-1}\T{B}{(AA^{+})}$ where conjugate gradient is used to compute product with the inverse. Note that conjugate gradient itself would have some error as we can't compute $\T{B}(I-AA^+)Bv$ exactly for a given vector $v$}.

Based on Theorem~\ref{thm:high-precision-regression}, we compute approximate projections onto the column span of $A$. The following lemma states that a matrix-vector product with the matrix $(\Delta/\beta^2)$ can be approximated well.

\begin{lemma}
	Given an arbitrary vector $v \in \Real^{d}$, we can compute a vector $y \in \Real^{d}$ such that
	\begin{equation*}
		\opnorm{y - (1/\beta^2){\Delta} v} \le \epsilonsub{reg}\kappa\opnorm{v}
	\end{equation*}
	in time $O(\nnz{B} + (\nnz{A}+c^2)\log(1/\epsilonsub{reg}))$.
	\label{lma:delta-by-beta-product}
\end{lemma}
\begin{proof}
	Recall that $\Delta = \T{B}(I-AA^+)B$. Therefore, for a vector $v$, $\Delta v = \T{B}Bv - \T{B}AA^+ Bv$. After computing $Bv$ exactly, we can compute $\tilde{y}$ by  Theorem~\ref{thm:high-precision-regression} in $O((\nnz{A}+c^2)\log(1/\epsilonsub{reg}))$ time such that
	\begin{equation*}
		\opnorm{AA^+ Bv - \tilde{y}} \le \epsilonsub{reg}\opnorm{(I-AA^+)Bv}.
	\end{equation*}
	Let $y = \T{B}Bv - \T{B}\tilde{y}$, which can be computed in $O(\nnz{B})$ time. Then $\Delta v - y = \T{B}(\tilde{y}-AA^+ Bv)$, which implies $\opnorm{\Delta v - y} \le \opnorm{B}\opnorm{\tilde{y} - AA^+ Bv} \le \epsilonsub{reg}\opnorm{B}\opnorm{(I-AA^+)B}\opnorm{v}$.
	
	Thus, given a vector $v$, we can compute $(\Delta/\beta^2)v$ up to an error of \[\epsilonsub{reg}\opnorm{B}\opnorm{(I-AA^+)B}\opnorm{v}/\beta^2 \le \epsilonsub{reg}\kappa\opnorm{v},\] since $\beta \ge \max(\opnorm{(I-AA^+)B},\sigma_{k+1}(B))$.
\end{proof}
\begin{lemma}
	Given an arbitrary vector $v \in \Real^d$, for matrix $M = r\left({\Delta}/{\beta^2}\right) = \sum_{j=0}^{t}r_j\left({\Delta}/{\beta^2}\right)^j$ where the degree $t = O(({1}/{\sqrt{\varepsilon}})\log({\kappa}/{\varepsilon}))$ and $\|r\|_1 = O((1+\sqrt{2})^{O(\sqrt{1/\varepsilon}\log(\kappa/\varepsilon))}\log(\kappa/\varepsilon))$, we can compute a vector $y$ such that $\opnorm{Mv - y} \le \epsilonsub{r}\opnorm{v}$ in time \[O\left(t \cdot \left(\nnz{B} + (\nnz{A} + c^2)\log\left({\kappa\|r\|_1}/{\epsilonsub{r}}\right)\right)\right).\]
\end{lemma}
\begin{proof}
Let $\Apx_0 := v$ and for $i \ge 1$, define $\Apx_i$ to be the approximation computed for the product $(\Delta/\beta^2)\Apx_{i-1}$ by Lemma~\ref{lma:delta-by-beta-product}. Define
\begin{equation*}
	E_i := \opnorm{(\Delta/\beta^2)^iv - \Apx_i}.
\end{equation*}
We have the following recurrence
\begin{align*}
	E_i = \opnorm{\left(\frac{\Delta}{\beta^2}\right)^iv - \Apx_i} &\le \opnorm{(\Delta/\beta^2)^iv - (\Delta/\beta^2)Apx_{i-1}} + \opnorm{(\Delta/\beta^2)\Apx_{i-1} - \Apx_i}\\
	&\le \opnorm{(\Delta/\beta^2)}E_{i-1} + \epsilonsub{reg}\kappa\opnorm{\Apx_{i-1}}\\
	&\le \opnorm{(\Delta/\beta^2)}E_{i-1}+ \epsilonsub{reg}\kappa \cdot (\opnorm{\Delta/\beta^2}^{i-1}\opnorm{v} + E_{i-1})\\
	&\le (\opnorm{\Delta}/\beta^2 + \epsilonsub{reg}\kappa)E_{i-1} + \epsilonsub{reg}\kappa\opnorm{\Delta/\beta^2}^{i-1}\opnorm{v}.
\end{align*}
As $\beta \ge (1+\varepsilon)\opnorm{(I-AA^+)B}$, we have that $\opnorm{\Delta/\beta^2} \le 1/(1+\varepsilon)^2$. If {$\epsilonsub{reg}\kappa \le \varepsilon/4$}, then $\opnorm{\Delta/\beta^2} + \epsilonsub{reg}\kappa \le 1/(1+\varepsilon)^2 + \varepsilon/4 \le 1/(1+\varepsilon)$. Therefore
\begin{equation*}
	E_i \le \frac{E_{i-1}}{1+\varepsilon}+ \frac{\epsilonsub{reg}\kappa}{(1+\varepsilon)^{2(i-1)}}\opnorm{v}.
\end{equation*}
This implies upon solving the recurrence that
\begin{equation*}
	E_i \le \epsilonsub{reg}\kappa\opnorm{v}
\end{equation*}
for all $i$.
Then 
\begin{align*}
	\opnorm{Mv - \sum_{j=0}^t r_j\Apx_j} &\le \sum_{j=0}^t|r_j|\opnorm{(\Delta/\beta^2)^jv - \Apx_j}\\
	 &\le \sum_{j=0}^t |r_j|E_j \le \epsilonsub{reg}\kappa\opnorm{v}\sum_{j=0}^t|r_j| = \epsilonsub{reg}\kappa\opnorm{v}\|r\|_1.
\end{align*}
So for any arbitrary vector $v$, we can compute a vector $y$ such that
\begin{equation*}
	\opnorm{Mv - y} \le \epsilonsub{r}\opnorm{v}
\end{equation*}
by setting $\epsilonsub{reg} = O(\frac{\epsilonsub{r}}{\kappa\|r\|_1}) \le \varepsilon/4\kappa$ for all $t$ approximate products and thus $y$ can be computed by Lemma~\ref{lma:delta-by-beta-product} in time
\begin{equation*}
	O(t \cdot (\nnz{B} + (\nnz{A} + r^2)\log\left({\kappa\|r\|_1}/{\epsilonsub{r}}\right))). 
\end{equation*}
This concludes the proof of the lemma.
\end{proof}
Thus for an arbitrary vector $v$, we can compute a vector $y$ such that $\opnorm{Mv - y} \le \epsilonsub{r}\opnorm{v}$. 

\begin{proof}[Proof of Lemma~\ref{lma:oracle-time-complexity}]
Recall that $\mathcal{M}' = (AA^+ BM)/\beta$, $\opnorm{M} \le 2/\sqrt{\varepsilon}$ and $\sigma_{\min}(M) \ge 1/2$ from Lemma~\ref{lma:replacing-neg-square-root}. We have $\opnorm{AA^+ BM} \ge \opnorm{AA^+ B}\sigma_{\min}(M) \ge \opnorm{AA^+ B}/2 \ge \varepsilon\opnorm{B}/2$ where the last ineqaulity follows from our assumption that $\opnorm{AA^+ B} \ge \varepsilon\opnorm{B}$. Thus $\opnorm{\mathcal{M}'} \ge \varepsilon\opnorm{B}/2\beta \ge \varepsilon/4$ as $\beta \le (1+\varepsilon)\opnorm{B}$.

Now we show how to compute approximations to $\mathcal{M}'v$ and $\T{\mathcal{M}'}v'$ for arbitrary vectors $v,v'$.

 To compute an approximation to $\mathcal{M'} v$, we first obtain a vector $y_1$ using the above lemma such that $\opnorm{Mv - y_1} \le \epsilonsub{r}\opnorm{v}$. Then we compute the product $By_1$ exactly in time $O(\nnz{B})$. Thereafter we compute a vector $y_2$ by Theorem~\ref{thm:high-precision-regression} such that
\begin{equation*}
	\opnorm{y_2 - AA^+ By_1} \le \epsilonsub{reg}\opnorm{(I-AA^+)By_1} \le \epsilonsub{reg}\opnorm{(I-AA^+)B}\opnorm{y_1}.
\end{equation*}
We also have $\opnorm{AA^+ By_1 - AA^+ BMv} \le \epsilonsub{r}\opnorm{AA^+ B}\opnorm{v}$. Therefore by the triangle inequality, $\opnorm{AA^+ BMv - y_2} \le \epsilonsub{r}\opnorm{AA^+ B}\opnorm{v} + \epsilonsub{reg}\opnorm{(I-AA^+)B}\opnorm{y_1}$. Hence
\begin{align*}
	\opnorm{\mathcal{M}' v - (y_2/\beta)} &\le \epsilonsub{r}\frac{\opnorm{AA^+ B}}{\beta}\opnorm{v} + \epsilonsub{reg}\opnorm{y_1} \le \epsilonsub{r}\kappa\opnorm{v} + \epsilonsub{reg}(\epsilonsub{r}\opnorm{v} + \opnorm{M}v)\\
	&\le \epsilonsub{r}(\kappa+1)\opnorm{v} + \frac{2\epsilonsub{reg}}{\sqrt{\varepsilon}}\opnorm{v}.
\end{align*}
Thus if $\epsilonsub{r} = O(\epsilonsub{f}\varepsilon/\kappa)$ and $\epsilonsub{reg} = O(\epsilonsub{f}\varepsilon^{3/2})$, we have that $\opnorm{\mathcal{M}'v - (y_2/\beta)} \le \epsilonsub{f}\varepsilon\opnorm{v} \le \epsilonsub{f}\opnorm{\mathcal{M}'}\opnorm{v}$. Therefore a vector $y_2/\beta$ can be computed in time $O(t \cdot (\nnz{B} + (\nnz{A} + c^2)\log\left(\frac{\kappa^2\|r\|_1}{\epsilonsub{f}\varepsilon}\right))) + O((\nnz{A} + c^2)\log(\frac{1}{\epsilonsub{f}\varepsilon}))$.

Now we compute an approximation to $\T{\mathcal{M}'}v = (\T{M}\T{B}AA^+/\beta) v$ for an arbitrary vector $v$. We first compute a vector $y_1$ such that 
\begin{equation*}
	\opnorm{AA^+ v - y_1} \le \epsilonsub{reg}\opnorm{(I-AA^+)v} \le \epsilonsub{reg}\opnorm{v}.
\end{equation*}
Then we compute $\T{B}y_1$ exactly. Then we compute a vector $y_2$ such that $\opnorm{M\T{B}y_1 - y_2} \le \epsilonsub{r}\opnorm{\T{B}y_1}\le \epsilonsub{r}\opnorm{B}(1+\epsilonsub{reg})\opnorm{v}$ .
We further have 
\begin{equation*}
	\opnorm{M\T{B}AA^+ v - M\T{B}y_1} \le \epsilonsub{reg}\opnorm{M\T{B}}\opnorm{v} \le \epsilonsub{reg}\frac{2\opnorm{B}}{\sqrt{\varepsilon}}\opnorm{v}.
\end{equation*}
Thus
\begin{equation*}
	\opnorm{y_2 - M\T{B}AA^+ v} \le 2\epsilonsub{r}\opnorm{B}\opnorm{v} + \epsilonsub{reg}\frac{2\opnorm{B}}{\sqrt{\varepsilon}}\opnorm{v}
\end{equation*}
and hence
\begin{equation*}
	\opnorm{y_2/\beta - \T{\mathcal{M}'} v} \le 2\epsilonsub{r}\kappa\opnorm{v} + \epsilonsub{reg}\frac{2\kappa}{\sqrt{\varepsilon}}\opnorm{v}
\end{equation*}
Now picking $\epsilonsub{r} = O(\epsilonsub{f}\varepsilon/\kappa)$ and $\epsilonsub{reg} = O(\epsilonsub{f}\varepsilon^{3/2}/\kappa)$, we obtain that 
\begin{equation*}
	\opnorm{(y_2/\beta) - \T{\mathcal{M}'} v} \le \epsilonsub{f}\varepsilon\opnorm{v} \le \epsilonsub{f}\opnorm{\mathcal{M}'}\opnorm{v}.
\end{equation*}
Thus, this approximation can be computed in time $O(t \cdot (\nnz{B} + (\nnz{A} + c^2)\log\left(\frac{\kappa^2\|r\|_1}{\epsilonsub{f}\varepsilon}\right))) + O((\nnz{A} + c^2)\log(\frac{\kappa}{\epsilonsub{f}\varepsilon}))$.
It follows that given an accuracy parameter $\epsilonsub{f}$, we can compute approximate matrix-vector products with $\mathcal{M}'$ and $\T{\mathcal{M}'}$ in time at most 
\begin{align*}
T(\epsilonsub{f}) &= O(t \cdot (\nnz{B} + (\nnz{A} + c^2)\log\left({\kappa(B)^2\|r\|_1}/{(\epsilonsub{f}\varepsilon})\right)))\\
&\quad + O((\nnz{A} + c^2)\log({\kappa(B)}/({\epsilonsub{f}\varepsilon}))).
\end{align*}
\end{proof}
\subsection{Proof of Theorem~\ref{thm:final-theorem}}\label{subsec:thm:final-theorem}
\begin{proof}
From Lemma~\ref{lma:replacing-neg-square-root}, $$\sigma_1(\mathcal{M}') \le \sigma_1\left(\frac{AA^+ B}{\beta}\right)\opnorm{M} \le \sigma_1\left(\frac{AA^+ B}{\beta}\right)\frac{2}{\sqrt{\varepsilon}}$$ and $$\sigma_{k+1}(\mathcal{M'}) \ge \sigma_{k+1}(AA^+ B/\beta)\cdot \sigma_{\min}(M) \ge \sigma_{k+1}(AA^+ B/\beta)(1/2).$$
Therefore $\kappa(\mathcal{M}') \le \sigma_1(AA^+ B/\beta)(2/\sqrt{\varepsilon})/\sigma_{k+1}(AA^+ B/\beta)/2 \le (4/\sqrt{\varepsilon})\kappa(AA^+ B)$. By Theorem~\ref{thm:main-theorem-krylov}, we can compute a matrix $Z \in \R^{n \times k}$ such that 
$
	\opnorm{(I-Z\T{Z})\mathcal{M'}} \le (1+2\varepsilon)\sigma_{k+1}(\mathcal{M}')
$
in time
\begin{equation*}
	T\left(\frac{\varepsilon}{\kappa(\mathcal{M}')^{5q}k^{11}C^q}\right)qk + T\left(\frac{\varepsilon^2}{48\kappa(\mathcal{M'}^2(\sqrt{qk})k)}\right)qk, 
\end{equation*}
where $q = O((1/\sqrt{\varepsilon})\log(d/\varepsilon))$. Thus the total time required is
\begin{equation*}
	O\left(tqk \cdot \left(\nnz{B} + (\nnz{A} + c^2)\log\left(\frac{\kappa^2\|r\|_1\kappa(\mathcal{M}')^{5q}k^{11} C^q}{\varepsilon^2}\right)\right)\right).
\end{equation*}
As $\|r\|_1 = (1+\sqrt{2})^{O(1/\sqrt{\varepsilon}\log(\kappa/\varepsilon))}\log(\kappa/\varepsilon)/\varepsilon$ and $\kappa(\mathcal{M}') = \kappa(AA^+ B)/\sqrt{\varepsilon}$, we obtain that the total time required is
$
	O(tqk\cdot \nnz{B} + tqk \cdot (\frac{1}{\sqrt{\varepsilon}}\log(\kappa/\varepsilon) + q)\log(\frac{\kappa \cdot \kappa(\mathcal{M}') \cdot k}{\varepsilon}) \cdot (\nnz{A} + c^2)).
$
Substituting $t = O(\sqrt{1/\varepsilon}\log(\kappa/\varepsilon))$, we obtain that the total running time is
\begin{equation}
	O\left(\left(\frac{\nnz{B} \cdot k}{\varepsilon} + \frac{\nnz{A}\cdot k}{\varepsilon^{1.5}} +\frac{c^2k}{\varepsilon^{1.5}}\right)\cdot \text{polylog}(\kappa, \kappa(AA^+ B),d,k,1/\varepsilon)\right), 
\end{equation}
and there is an additional $c^\omega$ time for computing a preconditioner.
%Thus we have the following Theorem.
By Lemmas~\ref{lma:apporximate-subspace-lra-rra} and \ref{lma:replacing-neg-square-root}, we obtain that
\begin{equation*}
    \opnorm{(AA^+ Z)(AA^+ Z)^+ B - B} = \opnorm{Z\T{Z}B - B} \le (1+O(\varepsilon))\opt.
\end{equation*}
The equality is from the fact that $Z$ is spanned by the columns of matrix $A$ by Theorem~\ref{thm:main-theorem-krylov}, and therefore $AA^+ Z = Z$. Thus, there exists a matrix $X_1 \in \R^{c \times k}$ such that $AX_1 = Z$ and the matrix $X_1$ can be computed in time $O((\nnz{A} + c^2)k + c^\omega)$ using sketching-based preconditioning techniques. Let $Y_1 = \T{Z}B$, which can be computed in time $O(\nnz{B}\cdot k)$. Therefore,
\begin{equation*}
	\opnorm{AX_1Y_1 - B} = \opnorm{Z\T{Z}B - B} \le (1+O(\varepsilon))\opt.
\end{equation*}
Thus $X_1 \cdot Y_1$ is a $(1+O(\varepsilon))$-approximation to the regression problem. By appropriately scaling $\varepsilon$, we obtain the proof.
\end{proof}
\subsection{Proof of Lemma~\ref{lma:removing-dependence}}\label{subsec:lma:removing-dependence}
\begin{proof}
Let $G \sim N(0,1)^{n \times (k+1)}$ and $\T{F} \in \R^{(k+1) \times d}$ be a matrix with $k+1$ orthonormal rows. Let $\alpha$ be a parameter to be chosen later and $\tilde{B} := B + \alpha G\T{F}$. For all matrices $X$, by the triangle inequality,
\begin{equation*}
    \opnorm{AX-\tilde{B}} \in \opnorm{AX - B} \pm \alpha\opnorm{G\T{F}}.
\end{equation*}
With probability $\ge 9/10$, $\opnorm{G} \le 2\sqrt{n}$. Thus
$
    \opnorm{AX-\tilde{B}} \in \opnorm{AX - B} \pm 2\alpha\sqrt{n}.
$
Therefore, if $\tilde{X}$ is a $(1+\varepsilon)$-approximation to $\min_{\text{rank-}k\ X}\opnorm{AX-\tilde{B}}$, then $\opnorm{A\tilde{X} - B} \le (1+\varepsilon)\opt + 6\alpha\sqrt{n}$.

We now have $\sigma_1(AA^+ \tilde{B}) \le \opnorm{\tilde{B}} \le \opnorm{B} + 2\alpha\sqrt{n}$ from the above discussion. We now lower bound $\sigma_{k+1}(AA^+ \tilde{B})$. Let $U$ be an orthonormal basis for the columns of $A$. Therefore $AA^+ = U\T{U}$.
\begin{align*}
    \sigma_{k+1}(AA^+ \tilde{B}) &= \sigma_{k+1}(U\T{U}\tilde{B})\\
    &= \sigma_{k+1}(\T{U}\tilde{B})\\
    &= \sigma_{k+1}(\T{U}B + \alpha\T{U}G\T{F})\\
    &\ge \sigma_{k+1}(\T{U}BF\T{F} + \alpha\T{U}G\T{F})\\
    &\ge \sigma_{k+1}(\T{U}BF + \alpha \T{U}G).
\end{align*}
As the rows of $\T{U}$ are orthonormal, the matrix $G' = \T{U}G$ is a matrix of i.i.d. normal random variables. Assuming $A$ is of full rank, $G'$ is a $c \times (k+1)$ matrix. Assuming $c \ge k+1$, let $E$ be the top $(k+1) \times (k+1)$ submatrix of $\T{U}BF + \alpha G'$. Then $E$ can be seen as a fixed $(k+1) \times (k+1)$ matrix where each entry is perturbed by a Gaussian random variable of variance $\alpha^2$. From Theorem~2.2 of \cite{tao-vu-perturbation}, we obtain that $\sigma_{\min}(E) \ge \alpha/(C\sqrt{k})$ for a constant $C$ with probability $\ge 9/10$. Thus $\sigma_{k+1}(\T{U}BF + \alpha\T{U}G) \ge \sigma_{\min}(E) \ge \alpha/(C\sqrt{k})$.

Thus, $\sigma_1(AA^+ \tilde{B})/\sigma_{k+1}(AA^+ \tilde{B}) \le (\opnorm{B} + 2\alpha\sqrt{n})/(\alpha/(C\sqrt{k}))$. For $\alpha = \frac{\varepsilon\sigma_{k+1}(B)}{(6\sqrt{n})}$, we obtain that
\begin{equation*}
    \sigma_1(AA^+ \tilde{B})/\sigma_{k+1}(AA^+\tilde{B}) \le \frac{Cn}{\varepsilon}\kappa
\end{equation*}
for a constant $C$ with probability $\ge 4/5$. Also, if $\tilde{X}$ is a $(1+\varepsilon)$-approximation as mentioned above, $\opnorm{A\tilde{X} - B} \le (1+\varepsilon)\opt + \varepsilon\sigma_{k+1}(B) \le (1+2\varepsilon)\opt$. We obtain the proof by scaling $\varepsilon$ appropriately.
\end{proof}